\documentclass[sn-mathphys,Numbered]{sn-jnl}% Math and Physical Sciences Reference Style
\usepackage{graphicx}%
\usepackage{multirow}%
\usepackage{amsmath,amssymb,amsfonts}%
\usepackage{amsthm}%
\usepackage{mathrsfs}%
\usepackage[title]{appendix}%
\usepackage{xcolor}%
\usepackage{textcomp}%
\usepackage{manyfoot}%
\usepackage{booktabs}%
\usepackage{algorithm}%
\usepackage{algorithmicx}%
\usepackage{algpseudocode}%
\usepackage{listings}%
\theoremstyle{thmstyleone}%

\newtheorem{theorem}{Theorem}[section]

\newtheorem{proposition}[theorem]{Proposition}
\newtheorem*{proposition*}{Proposition}
\newtheorem{definition}[theorem]{Definition}
\newtheorem*{definition*}{Definition}
\newtheorem{lemma}[theorem]{Lemma}
\newtheorem*{lemma*}{Lemma}
\newtheorem{corollary}[theorem]{Corollary}
\newtheorem*{corollary*}{Corollary}
\theoremstyle{thmstyletwo}%

\theoremstyle{thmstylethree}%
\usepackage{amsfonts,amsthm,amsmath,amssymb,mathrsfs,epsfig,color,ulem,comment} 
\usepackage{hyperref}
\allowdisplaybreaks[3]
\newcommand{\be}{\begin{equation}}
\newcommand{\ee}{\end{equation}}
\newcommand{\ba}{\begin{eqnarray}}
\newcommand{\ea}{\end{eqnarray}}

\newcommand{\Cr}[1]{{{\cal O} \left( r^{#1} \right)}}
\newcommand{\Crs}[1]{o\left(r_0^{~#1}\right)}
\newcommand{\Crz}[1]{{{\cal O} \left( r_0^{~#1} \right)}}

\newcommand{\er}[1]{{\textcolor{red}{\sout{#1}}}}
\newtheoremstyle{definition}{6pt}{6pt}{\rm}{}{\sffamily}{ }{ }{}
\DeclareRobustCommand{\er}{\bgroup\markoverwith{\textcolor{red}{\rule[.5ex]{2pt}{0.8pt}}}\ULon}

\begin{document}

\title[Article Title]{A generalization of photon sphere based on escape/capture cone}


\author[1,2]{\fnm{Masaya} \sur{Amo}}

\author[3,4]{\fnm{Keisuke} \sur{Izumi}}

\author[5,6]{\fnm{Hirotaka} \sur{Yoshino}}

\author[7]{\fnm{Yoshimune} \sur{Tomikawa}}

\author[4,3]{\fnm{Tetsuya} \sur{Shiromizu}}

\affil[1]{\orgdiv{Center for Gravitational Physics and Quantum Information, Yukawa Institute for Theoretical Physics}, \orgname{Kyoto University}, \orgaddress{\postcode{606-8502}, \state{Kyoto}, \country{Japan}}}

\affil[2]{\orgdiv{Departament de F{\'\i}sica Qu\`antica i Astrof\'{\i}sica, Institut de
Ci\`encies del Cosmos}, \orgname{Universitat de
Barcelona}, \orgaddress{\postcode{E-08028}, \state{Barcelona}, \country{Spain}}}

\affil[3]{\orgdiv{Kobayashi-Maskawa Institute}, \orgname{Nagoya University}, \orgaddress{\postcode{464-8602}, \state{Nagoya}, \country{Japan}}}

\affil[4]{\orgdiv{Department of Mathematics}, \orgname{Nagoya University}, \orgaddress{\postcode{464-8602}, \state{Nagoya}, \country{Japan}}}

\affil[5]{\orgdiv{Department of Physics}, \orgname{Osaka Metropolitan University}, \orgaddress{\postcode{558-8585}, \state{Osaka}, \country{Japan}}}

\affil[6]{\orgdiv{Nambu Yoichiro Institute of Theoretical and Experimental Physics (NITEP)}, \orgname{Osaka Metropolitan University}, \orgaddress{\postcode{558-8585}, \state{Osaka}, \country{Japan}}}

\affil[7]{\orgdiv{Division of Science, School of Science and Engineering}, \orgname{Tokyo Denki University}, \orgaddress{\postcode{350-0394}, \state{Saitama}, \country{Japan}}}


\abstract{In general asymptotically flat spacetimes, bearing the null geodesics reaching the future null infinity in mind, 
we propose new concepts, the ``dark horizons'' (outer dark horizon and inner dark horizon)  as generalizations of the photon sphere.
They are defined in terms of the structure
of escape/capture cones of photons
with respect to a unit timelike vector field to capture the motion of light sources.
More specifically, considering a two-sphere
that represents a set of emission directions of photons,
the dark horizons are 
located at positions where a hemisphere is marginally included in the
capture and escape cones, respectively.
In addition, our definition succeeds in incorporating relativistic beaming effect.
We show that the dark horizon is absent in the Minkowski spacetime,
while they exist in spacetimes with black hole(s) under a certain condition. We derive the general properties of the dark horizons in spherically
symmetric spacetimes and 
explicitly calculate the locations of the dark horizons in the Vaidya spacetime and the Kerr spacetime. In particular, in the Kerr spacetime, the outer dark horizon coincides with the shadow observed from infinity on the rotation axis. }




\maketitle

%
%
%======================================%
%<<<<<<<<<<<< SECTION I  >>>>>>>>>>>>>>%
%======================================%
%
\section{Introduction}
\label{Sec:intro}

Recently, Event Horizon Telescope (EHT) Collaboration succeeded in obtaining the image of the regions around black holes in M87~\cite{Akiyama:2019cqa} and in Sagittarius A*~\cite{EventHorizonTelescope:2022xnr}.
Future observations of the black hole shadow are expected to test the
general relativistic magnetohydrodynamics
(GRMHD) models~\cite{EventHorizonTelescope:2019pgp}, to constrain theories of gravity~\cite{Wang:2018prk,Moffat:2019uxp,Khodadi:2020gns}, and so forth.
In static and spherically symmetric spacetimes, 
the edge of the shadow is given by the photon sphere under the assumption that light sources are distributed at infinity~\cite{Virbhadra:1999nm,Claudel:2000}, which brings the photon sphere to attention in recent years. We note that the realistic observed shadow depends strongly on the distribution of the light source~\cite{Gralla:2019xty}, but the photon sphere is helpful to understand the geometric structures near the black hole.
Following the definition of Ref.~\cite{Claudel:2000}, 
the photon sphere is provided as
a $SO(3)\times {\mathbb{R}}$-invariant photon surface
in a static and spherically symmetric spacetime.
Here, the photon surface is defined as a nowhere-spacelike hypersurface $S$ such that, any null geodesic tangent to $S$ at any point on $S$ is included in $S$. 
Interestingly, the outermost photon sphere has been shown to satisfy the areal inequality~\cite{Yang:2019zcn} (see also~\cite{Lu:2019zxb})
\ba
\frac{9}{4}A_H\le A_{\rm ph,out}\le A_{\rm sh,out} \le 36\pi M^2,\label{AreaInequality}
\ea
where $A_H$ is the area of the event horizon, $A_{\rm ph,out}$ is the area of the outermost photon sphere, $A_{\rm sh,out}$ is the area of the shadow observed at infinity, and $M$ is the Arnowitt-Deser-Misner (ADM) mass.

However, the definition of the photon sphere does not work in spacetimes with the rotation. In addition, since the observed intensity depends on the position and the motion of the light source, it would be more useful if we could take these into account in the definition. There have been several attempts to generalize the photon sphere~\cite{Shiromizu:2017ego,Yoshino:2017gqv,Cao:2019vlu,Yoshino:2019dty}, but most of them are defined with local geometrical quantities. The escape probability of photons to infinity (or sufficiently far away from black holes) is not very clear in such definitions because the global analysis for the geodesic is crucial for each photon coming all the way from light sources around black holes. In this sense, we will focus on generalizations of the photon sphere with an attention to the global properties, especially the escape probability of photons, associated with the motion of the light sources.
So far, at least to 
our knowledge, there have been no generalized photon spheres that incorporate the effect of the position and the motion of the light source~\footnote{In Refs.~\cite{Siino:2019vxh,Siino:2021kep}, the approach based on the existence of an infinite number of
conjugate points along a null geodesic congruence is adopted. 
However, these references do not take account of the effect of the motion of the light source.}.

In this paper, we introduce new generalizations of the photon sphere
in terms of the escape probability of photons from the light sources incorporating the effect of its motion and the beaming effect due to their motion, which we name the {\it outer dark horizon (ODH)} and the {\it inner dark horizon (IDH)}.
These new concepts are defined by using the escape/capture cones, which depends on the frame associated with a unit timelike vector $T^\mu$.
We take $T^\mu$ as the tangent vector field of the trajectory of the light source (in the region where light sources exist). Together with this $T^\mu$, the escape/capture cones tell us how brightly the source is observed from infinity. 
Note that we can deal with the situation where the light sources are possibly distributed everywhere in the spacetime~\footnote{In this point, our definition is different from the wandering geodesic defined in Refs.~\cite{Siino:2019vxh,Siino:2021kep}
in which the light sources are supposed to be located at past null infinity.}.
Note that all of the rigorous theorems/propositions/lemmas are irrelevant to how to choose $T^\mu$ unless specified.

The rest of this paper is organized as follows.  In Sec.~\ref{Sec:review}, we provide a brief review of the asymptotic behavior 
of null geodesics near future null infinity as a preparation.
In Sec.~\ref{Sec:DH}, we introduce the ODH and IDH as generalizations of the photon sphere in general asymptotically flat spacetimes and probe the
properties of the ODH and IDH to see that the definition is fairly reasonable.
In Sec.~\ref{Sec:stsp}, we focus on a static and spherically symmetric case as a specific example and study basic properties of the ODH and IDH, mainly focusing on the relation to the photon sphere.
In Sec.~\ref{Sec:Vaidya}, we see their properties in the Vaidya spacetime as a simple example of dynamical spacetimes.
In Sec.~\ref{Sec:Kerr}, we investigate the
ODH and IDH in the Kerr spacetime as an example of spacetimes
with angular momenta.
The last section gives a summary and discussions.
In App.~\ref{Sec:Proof_convex-upward}, we provide the derivation of the inequalities appeared in Sec.~\ref{Sec:Kerr}.
We use the units in which the speed of light and the Newtonian constant of gravitation are unity, $c=1$ and $G=1$.
We assume the metric to be $C^{2-}$ functions ({\it i.e.}, class $C^{1,1}$).
The metric sign convention is $(-,+,+,+)$.

%
%
%======================================%
%<<<<<<<<<<<< SECTION II  >>>>>>>>>>>>>>%
%======================================%
%

\section{Review of Null Asymptotics in the Bondi coordinate} 

\label{Sec:review}

In this section, as a preparation for
the proof of 
Lem.~\ref{OLP} and Lem.~\ref{ILP}, we review the asymptotic behavior of the metric and null geodesics near future null infinity in four-dimensional asymptotically flat spacetimes.
The detailed analyses are presented in Refs.~\cite{Amo:2021gcn,Amo:2021rxr,Amo:2022tcg,Amo:2023qws}.

We begin with reviewing the Bondi coordinates based on Refs.~\cite{Bondi,Sachs}, which describes the asymptotic behavior of the metric near future null infinity (for the generalization to higher dimensions, see also Refs.~\cite{,Tanabe:2011es,Hollands:2003ie,Hollands:2003xp,Ishibashi:2007kb}).
The non-zero components of the metric near future null infinity in the Bondi coordinates can be expanded in the power of $1/r$ as
\begin{eqnarray}
  g_{uu}&=&-1+mr^{-1} +\Cr{-2},\label{metricuu} \\
  g_{ur} &=&-1 + \Cr{-2},\label{metricur}\\ 
  g_{IJ} &=&h_{IJ}r^2~=~\omega_{IJ}r^2 + h^{(1)}_{IJ}r + \Cr{0},  \label{metricIJ}\\
  \quad g_{uI}&=& \Cr{0}.\label{metricuI}
  \end{eqnarray}
Here, $u$ denotes the retarded time, $r$ is the areal radius, $x^I$ stands for the angular coordinates, and $\omega_{IJ}$ is the metric for the unit two-sphere.  
The non-zero components of the inverse metric behave as 
\begin{eqnarray}
g^{ur} &=&  -1 + \Cr{-2},\label{g^ur}\\
g^{rr}  &=&1-mr^{-1}+\Cr{-2},\label{g^rr}\\
g^{rI} &=& \Cr{-2},\\
g^{IJ} &=&   \omega^{IJ}r^{-2} -h^{(1)IJ}r^{-3} + \Cr{-4},\label{g^IJ}
\end{eqnarray}
where $\omega^{IJ}$ is the inverse of $\omega_{IJ}$, and $h^{(1)IJ}$ is defined as $h^{(1)IJ}:=\omega^{IK}\omega^{JL}h^{(1)}_{KL}$.
Future null infinity is supposed to be in the limit of $r \to \infty$ while $u$ is kept finite. 
Let us impose the gauge condition   
\begin{align}
    \label{gau}
    \sqrt{\det h_{IJ}}=\omega_{2},
  \end{align}
where $\omega_{2}$ is the volume element of the two-dimensional unit sphere.
Here, $h_{IJ}-\omega_{IJ}$ corresponds to gravitational waves.
In general relativity, the integration of $m(u,x^I)$ over the angular coordinates gives us the Bondi mass, 
\begin{eqnarray}
M(u):=\frac{1}{8\pi}\int_{S^{2}}md\Omega \label{M(u)}.
\end{eqnarray}

Next, let us look at the asymptotic behavior of null geodesics near future null infinity based on Refs.~\cite{Amo:2021gcn,Amo:2021rxr,Amo:2022tcg,Amo:2023qws}.
By combining the $r$-component of the geodesic equation and the condition for the geodesic to be null, we have
\ba
r''&=&\Omega_{IJ}r\left(x^I\right)'\left(x^J\right)'+\Cr{0}\left|\left(x^I\right)'\right|^2+\Cr{-2}r'^2,\label{r''}
\ea
where $\Omega_{IJ}$ is defined as 
\ba
\Omega_{IJ}:=\omega_{IJ}  - \frac12 \frac{\partial h^{(1)}_{IJ}}{\partial u} + \frac12 \frac{\partial m}{\partial u} \omega_{IJ}.
\ea
Here, the prime denotes the derivative with respect to the affine parameter.
We also define $\Omega\left(u, {x^I};dx^J/du\right)$ and $\Omega_i$ as 
\ba
\Omega\left(u, {x^I};\frac{dx^J}{du}\right)&:=&\left(\omega_{IJ} \frac{dx^I}{du}\frac{dx^J}{du}\right)^{-1}\Omega_{KL} \left(u, {x^M}\right) \frac{dx^K}{du} \frac{dx^L}{du},\label{def_Omega}\\
\Omega_i&:=&\inf_{u, {x^I}, dx^J/du}\Omega\left(u, {x^I};\frac{dx^J}{du}\right).\label{defOi}
\ea
In Refs.~\cite{Amo:2021gcn,Amo:2022tcg,Amo:2023qws},
the careful analysis of the global behavior of null geodesics with Eq.~\eqref{r''} gives us a sufficient condition to reach future null infinity.
In particular, the result for $\Omega_i>0$, which we use later, is given by the following statement:

\begin{proposition}
  \label{PropIV}
Consider a four-dimensional asymptotically flat spacetime with $\Omega_i>0$ in which the metric near future null infinity is written as Eqs.~\eqref{metricuu}--\eqref{metricuI} with the Bondi coordinates by $C^{2-}$ functions.
We define $\beta_{\rm crit}$ as
\ba
  \beta_{\rm crit}:=
  \frac{-3+\sqrt{9-6\Omega_i}}{3}.
\ea 
Take a point $p$ with a sufficiently large radial coordinate value $r=r_0$.
Any null geodesic emanating from $p$ reaches future null infinity if 
\ba
0<\left(\left.\frac{dr}{du}\right|_p-\beta_{\rm crit}\right)^{-1}&=&\Crs{}\label{con1}
\ea
holds.
\end{proposition}
Here, Eq.~\eqref{con1} roughly means $dr/du|_p\gtrsim\beta_{\rm crit}$.
In the case with $\partial m/\partial u$ and $\partial h^{(1)}_{IJ}/\partial u$ being small enough, which corresponds to the situation that matter radiations and gravitational waves are both weak enough near future null infinity, the value of $\beta_{\rm crit}$ is approximately given by $\beta_{\rm crit}\approx-(1-1/\sqrt{3})\approx-0.423$.
By this proposition, we see that photons emitted with $dr/du\ge0$ and ones with $dr/du=\Crz{-1}$ reach future null infinity under the assumptions in the proposition.
For higher dimensions, see Refs.~\cite{Amo:2021gcn,Amo:2022tcg}.

%
%
%======================================%
%<<<<<<<<<<<< SECTION III  >>>>>>>>>>>>>>%
%======================================%
%
\section{Outer Dark Horizon and Inner Dark Horizon}

\label{Sec:DH}

In this section, we introduce two new concepts, ``outer dark horizon (ODH)'' and ``inner dark horizon (IDH),'' as  generalizations of the photon sphere in general asymptotically flat spacetimes. 
The definitions of the ODH and IDH refer to future null infinity as the event horizon does.

This section is organized as follows.
In subsection~\ref{Subsec:basicsDH}, we present the definitions of the ODH and IDH, and discuss their properties. 
In subsection~\ref{Subsec:twocones}, the relation of the ODH and IDH to the escape and capture cone is discussed.
In subsections~\ref{Subsec:PropertyODH} and \ref{Subsec:PropertyIDH}, we examine the properties of the ODH and IDH, respectively.

\subsection{Basics of Outer Dark Horizon and Inner Dark Horizon}
\label{Subsec:basicsDH}

In this subsection, we introduce the outer dark horizon (ODH) and the inner dark horizon (IDH) as generalizations of the photon sphere by examining whether
photons emitted from photon sources reach future null infinity or not. 
The ODH and IDH are defined as the boundary of ``outer dark domain (ODD)'' and ``inner dark domain (IDD)'', respectively, which are also defined below.
 
Let ${\mathcal M}$ be a four-dimensional asymptotically flat and connected spacetime, and $T^\mu$ be a unit timelike vector field with $T_\mu T^\mu = -1$~\footnote{The reason why we normalized $T^\mu$ is that only the direction of $T^\mu$, but not its norm, is important in Defs.~\ref{ODD} and \ref{IDD}.
Note that the norm of $n^\mu$ in the following deﬁnitions is also not important.}. 

\begin{definition}[{\bf outer dark domain} ({\bf ODD}) {\bf associated with $T^\mu$}]
  \label{ODD}
  Suppose a set $S_O$ consists of all $p\in \mathcal{M}$ satisfying the following condition: for all spacelike vectors $n^\mu$ orthogonal to $T^\mu$ at $p$, there exists a null geodesic emanating from $p$, whose tangent vector $k^\mu$ is orthogonal to $n^\mu$ at $p$, such that it will not reach future null infinity.
  Then, we call $S_O$
  the outer dark domain (ODD) associated with $T^\mu$.
\end{definition}

\begin{definition}[{\bf outer dark horizon} ({\bf ODH}) {\bf associated with $T^\mu$}]
  \label{ODH}
  Let a region $S_O$ be the ODD associated with $T^\mu$ in ${\mathcal M}$. Then, the boundary $\partial S_O$ is called the outer dark horizon (ODH) associated with $T^\mu$.
\end{definition}

\begin{definition}[{\bf inner dark domain} ({\bf IDD}) {\bf associated with $T^\mu$}]
  \label{IDD}
  Suppose a set $S_I$ consists of all $p\in \mathcal{M}$ satisfying the following condition: there exists a spacelike vector $n^\mu$ orthogonal to $T^\mu$ at $p$, such that all null geodesics emanating from $p$ whose
  tangent vector $k^\mu$ is orthogonal to $n^\mu$ at $p$ will not reach future null infinity $\mathscr{I}^+$.
    Then, we call $S_I$
    the inner dark domain (IDD) associated with $T^\mu$.
\end{definition}

\begin{definition}[{\bf inner dark horizon} ({\bf IDH}) {\bf associated with $T^\mu$}]
  \label{IDH}
  Let a region $S_I$ be the IDD associated with $T^\mu$ in ${\mathcal M}$. 
  Then, the boundary $\partial S_I$ is called the inner dark horizon (IDH) associated with $T^\mu$.
\end{definition}

The outer dark domain and the inner dark domain are sometimes simply called the dark domains in short.
Similarly, the outer dark horizon and the inner dark horizon are simply called the dark horizons.
The ODD and IDD are complementary concepts to each other in the sense that the roles of ``there exists'' and ``for all'' are switched in the definitions.
We required the vector $n^\mu$ to be orthogonal to $T^\mu$ in the definitions of \ref{ODD} and \ref{IDD} in order to consider the escape condition of photons in the frame associated with $T^\mu$.
As a result, the locations of dark domains and dark horizons depend on $T^\mu$.
We take $T^\mu$ as tangent vectors of the light source worldlines. 
The escape condition of the emitted photon depends on the motion of the sources due to the effect of the relativistic beaming. 
Here, it is also possible to consider situations where $T^\mu$ is defined only in a subregion of the spacetime and define the dark domain and dark horizon in this subregion, but we mainly focus on situations where $T^\mu$ is defined globally in the whole spacetime.

In addition, we can consider a particular class of the outer/inner dark domain/horizon.
Suppose that a time coordinate $t$ is present, and a spacelike
hypersurface $\Sigma_t$ is given by $t=\mathrm{constant}$.
From the time coordinate $t$,
we can naturally introduce a timelike unit vector field by
\begin{equation}
T^\mu \ =\ \frac{-\nabla^\mu t}{\sqrt{-\nabla^\nu t\nabla_\nu t}}.
\end{equation}
Adopting this vector field as $T^\mu$ in the definitions
of the ODD and IDD, 
we can introduce outer/inner dark domain/horizon associated with
the time slice $\Sigma_t$ as follows:\footnote{
  This is a particular case where $T^\mu$ satisfies
$T_{[\mu}\nabla_{\nu}T_{\rho]}=0$ by Frobenius's theorem.} 

\begin{definition}[{\bf outer/inner dark domain/horizon associated with $\Sigma_t$}]
  \label{DD_DH_associated_with_time_slice}
  Suppose a time coordinate $t$ is present, and the spacelike hypersurface
  $t=$constant is denoted as $\Sigma_t$. 
  The ODD/IDD and the ODH/IDH associated with
  the timelike vector field $T^\mu=-\nabla^\mu t/\sqrt{-\nabla^\nu t\nabla_\nu t}$
  are called the ODD/IDD and the ODH/IDH associated with time slice $\Sigma_t$.
\end{definition}

For the outer/inner dark domain/horizon associated with $\Sigma_t$,
the vector $n^\mu$ in the definitions of \ref{ODD} and \ref{IDD}
is tangent to $\Sigma_t$. 
In Sec.~\ref{Sec:stsp}, we will see that the ODH and IDH are generalizations of the photon sphere.
We will also see that they are both absent in the Minkowski spacetime and both exist in a spacetime with black hole(s) under certain conditions. These facts imply that the existence of the ODH/IDH is expected to be an indicator for strong gravity.
The words ``outer'' and ``inner'' refer to the fact
that the IDD is generally included by the ODD, which
will be confirmed in Cor.~\ref{IDD-is-included-in-ODD}.

Since the geodesic equations are second-order differential equations, the initial position and its first derivative provide us with the unique solution~\footnote{This is because we have assumed that the metric is $C^{2-}$ functions ({\it i.e.}, class $C^{1,1}$).}.
Therefore, at any point in the spacetimes, the emission angle of a photon
(in the frame of the photon source) 
determines whether it reaches future null infinity or not.

\subsection{Relation to escape and capture cones}
\label{Subsec:twocones}
The escape cone is a useful concept to describe the condition for photons
to reach future null infinity.
In order to explain this concept, 
we introduce the tetrad basis, $T^\mu$ and $(e_i)^\mu$ with $i=1,2$, and $3$,
and the projection tensor $\gamma_{\mu\nu}=g_{\mu\nu}+T_{\mu}T_{\nu}$
onto the three-dimensional spacelike tangent subspace
spanned by $(e_i)^\mu$.
Then, the emission direction
in the frame of the photon source is expressed by the tangent null vector $k^\mu$ of a photon
\begin{equation}
  E^\mu \, =\, \frac{k_\perp^\mu}{|k_\perp^\nu|}, \quad \textrm{with}
  \quad k_\perp^\mu\, =\, {\gamma^\mu}_{\nu}k^\nu.
  \label{emission-direction-vector}
\end{equation}
Since $E^\mu$ is a unit vector, there is the apparent
one-to-one correspondence between $E^\mu$
and a point of a unit two-sphere in the tangent subspace, which
we call the ``two-sphere of emission directions.'' 
For each point on the two-sphere,
we can judge whether the corresponding photon
escapes to future null infinity or not.
As a result, the two-sphere
is divided into the escape region and the capture region.
If we connect the central point of the two-sphere and the points
on the boundary of these two regions with straight lines,
the escape and capture cones appear.
Below, we use the words the ``escape cone'' and the ``capture cone''
with the same meanings as the
escape region and the capture region on the two-sphere
of emission directions, respectively.

Null geodesics that correspond to the boundary of the escape and
capture cones coincide with the future wandering null geodesics
defined by Siino \cite{Siino:2019vxh,Siino:2021kep}.
Since the wandering null geodesics never arrive at future
null infinity, the boundary is regarded to belong to the capture cone. Therefore, the capture cone is a
closed set, while the escape cone is an open set.

\begin{figure}[htbp]
  \begin{center}
  \includegraphics[width=3.2cm]{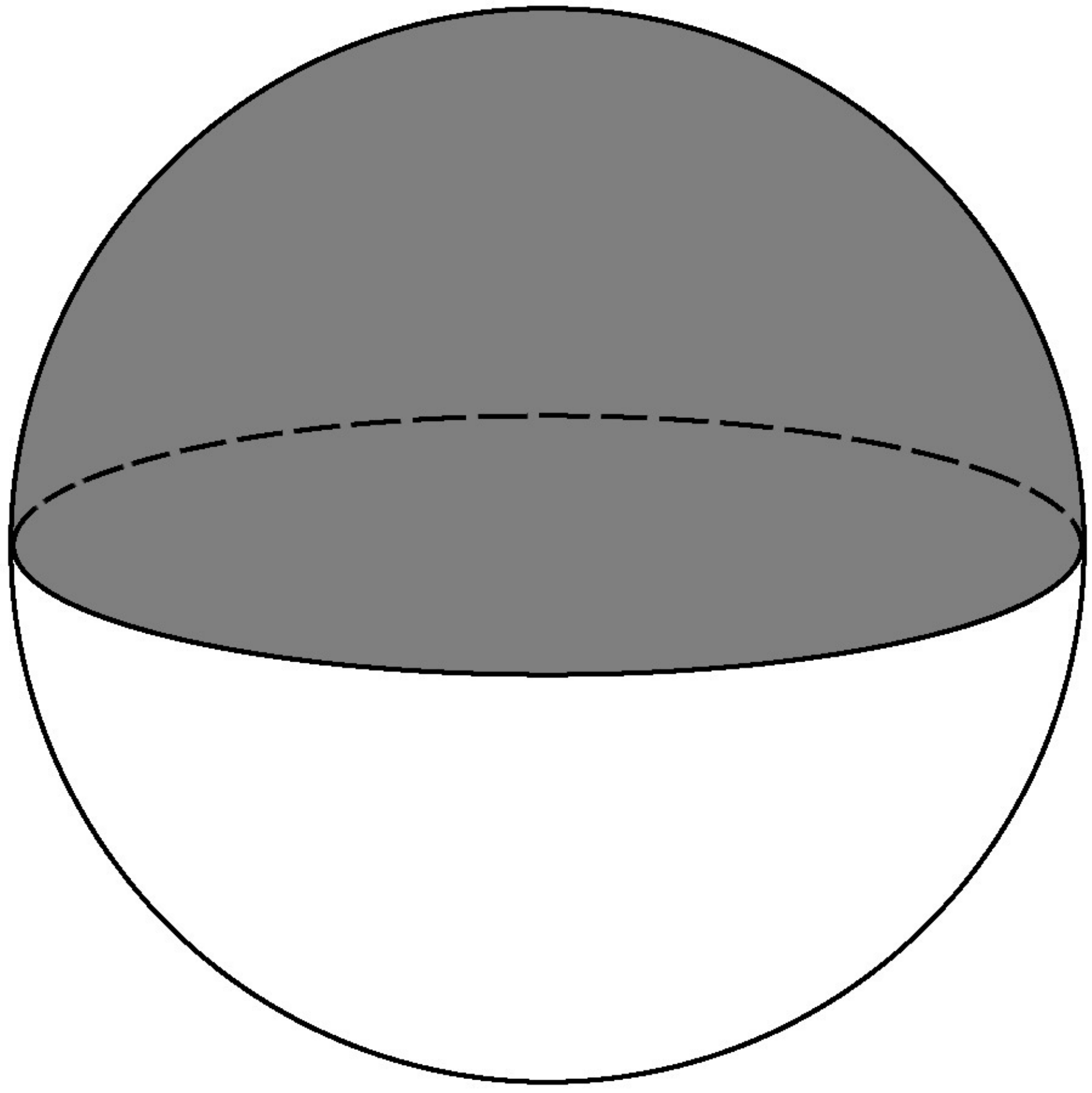}
  \caption{Representation of the two-sphere of emission directions at $r=3M$ in the Schwarzschild spacetime, at which the photon sphere is located. 
    The escape region 
    and the capture region are shown by
    the grey region and the white region, respectively. The boundary of the two regions is given by an orthodrome.
  }
  \label{Sch_escape}
  \end{center}
  \end{figure}

In a static and spherically symmetric spacetime, the boundary of an escape region associated with a constant-time slice is given by an exact circle. In particular, on the photon sphere in the Schwarzschild spacetime, the boundary of the escape region is given by an orthodrome of the two-sphere (see Fig.~\ref{Sch_escape}).
One might think that it would be reasonable to define a generalization of the photon sphere as a set of points at which the boundary of the escape region is an orthodrome associated with a constant-time slice, at least for stationary spacetimes.
However, in general spacetimes, the boundary of the escape region is not necessarily an exact circle.
This leads us to the physical interpretation
for our new concepts as described below.

  To characterize emission directions $E^\mu$ orthogonal to a spacelike vector $n^\mu$ in (negations of) Defs.~\ref{ODD} and \ref{IDD}, the notion of  orthodrome is useful.
  In the tangent subspace orthogonal to a unit timelike vector field $T^\mu$, an orthodrome is defined as an intersection between a plane containing the origin of the tangent subspace and the two-sphere of emission directions (see circles on the two-spheres in Fig.~\ref{fig:DH}).
  Note that an orthodrome orthogonal to $n^\mu$ represents a set of emission directions $E^\mu$  orthogonal to $n^\mu$.

  The conditions for the ODD and IDD associated with $T^\mu$
  are now rephrased using 
  the escape and capture cones as follows. 
  On the one hand, the negation of Def.~\ref{ODD} is equivalent
  to that 
  at the spacetime point $p$ outside the ODD,
  there exists an orthodrome that is entirely included in the
  escape cone.
  On the other hand, Def.~\ref{IDD} is equivalent
  to that at a spacetime point $p$ in an IDD,
  there exists an orthodrome that is entirely included in the
  capture cone. 
  Note that each of the escape and capture cones can have
  multiple components.
  For example, at a point between two distant black holes,
  there would be two capture cones and one ``escape belt''.
  In the case that each of the escape and capture cones 
  has a single component, a hemisphere is included in the
  capture cone at a point in the IDD,
  and a hemisphere is included in the escape cone at
  a point outside the ODD.

  Figure~\ref{fig:DH} illustrates the situation
  where each of the escape and capture cones has a single component.
  Spacetime points can be divided into three classes. 
  In the first class (the left sphere in Fig.~\ref{fig:DH}),
  there exists an orthodrome which is totally included by the capture cone,
  and the corresponding point in the spacetime is in the IDD.
  This point is also in the ODD because
  the escape cone is included in a hemisphere and it cannot
  include an orthodrome 
  (a strict proof
  will be given in Cor.~\ref{IDD-is-included-in-ODD}).
  In the second class (the center sphere in Fig.~\ref{fig:DH}),
  there is no orthodrome that is entirely
  included by the escape region or by the capture region. 
  The corresponding points in the spacetime are in the ODD, but not in the IDD. 
  In the last class (the right sphere in Fig.~\ref{fig:DH}), there exists a hemisphere which is totally included in the escape region. In this case,
  the corresponding point is not in the ODD as mentioned above.
  In addition, since the capture cone is included
  in a hemisphere, 
  this point is not in the IDD
  (see a strict proof for Cor.~\ref{IDD-is-included-in-ODD}).

  In short, 
    any point in the ODD has a small escape cone, and any point not in the ODD has a large escape cone. 
    As for the IDD, 
    any point in the IDD has a large capture cone, and any point not in the IDD has a small capture cone. 
   There is no point such that it is in the ODD and not in the IDD, which will be proved more strictly in Cor.~\ref{IDD-is-included-in-ODD}.  

   \vspace{1.5mm}
   \begin{figure}[htbp]
         \begin{center}
           \includegraphics[width=12cm]{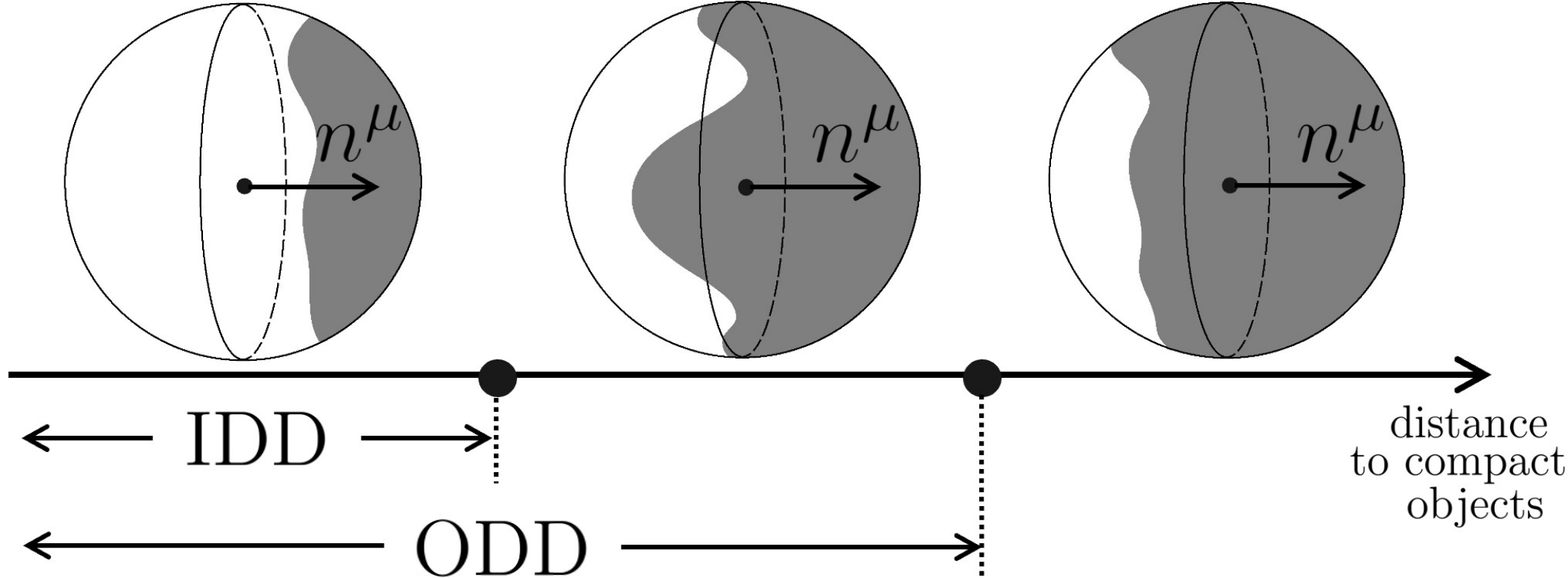}
             \end{center}
             \caption{Depiction of two-spheres of emission directions, which represent the escape regions (the grey regions) and the capture regions (the white regions) on three points in a spacetime without symmetry. A vector orthogonal to the plane including the orthodrome corresponds to the spacelike vector $n^\mu$ in (negations of) Defs.~\ref{ODD} and \ref{IDD}.  {\it Left sphere:} An illustration of the escape and capture regions at a certain point. In this case, this point is both in the ODD and in the IDD.  {\it Center sphere:} A picture of the escape and capture regions at another point. In other words, all orthodromes intersect the boundary of the two regions. In this case, the point is in the ODD but not in the IDD. {\it Right sphere:} An illustration of the escape and capture regions at another point. In this case, the point is not in the ODD nor IDD. 
             }
   \label{fig:DH}
   \end{figure}

   We now present several useful propositions
   in order to judge whether a given spacetime point
   is in an ODD/IDD or not.
   For this purpose, we prove the following Lemma:
  \begin{lemma}
    \label{unit-sphiere-antipodes}
    Consider a unit two-sphere and a pair of antipodal points
    $P$ and $Q$ on that sphere. Let $\gamma$ be any continuous curve
    that connects $P$ and $Q$, and $\gamma$ includes the
    two endpoints $P$ and $Q$. Then,
    any orthodrome on that sphere intersects $\gamma$.
  \end{lemma}
  \begin{proof}
    In the case that 
    the orthodrome includes the points $P$ and $Q$,
    it trivially intersects $\gamma$.
    In the case that the orthodrome does not include $P$ and $Q$,
    the orthodrome divides the two-sphere into two regions, and the points
    $P$ and $Q$ are included in different regions because they are antipodal points. Since the curve $\gamma$
    continuously connects $P$ and $Q$, it must cross the boundary
    of the two regions. That is, the orthodrome intersects $\gamma$.
  \end{proof}  
  Suppose that a pair of antipodal points are included in a connected component of the capture cone at some spacetime point.
   Then, 
   there exists a curve that connects the two points
   in the capture cone. From the above Lem.~\ref{unit-sphiere-antipodes},
   any orthodrome crosses that curve,
   and hence, it cannot be included entirely in the escape cone.
   This means that such a spacetime point is in an ODD. For a spacetime point at which each of the escape and capture cones is path-connected, we easily see that the reverse is also true and obtain the following:
\begin{theorem}
  \label{antipodal-points-capture-cone-ODD}
  If a pair of antipodal points are included in a connected component of the capture cone, the point in the spacetime is in an ODD.
  In particular, suppose that each of the escape and capture cones is path-connected. Then, such a spacetime point is in an ODD if and only if the capture cone contains a pair of antipodal points.
  \end{theorem}  

Similarly, when a pair of antipodal points are included in a connected component of the escape cone, 
   there exists a curve that connects the two points
   in the escape cone. Since any orthodrome crosses that curve,
   it cannot be entirely included in the capture cone.
   This means that such a spacetime point is not in an IDD. For a point at which each of the escape and capture cones is path-connected, we see that the reverse also holds.
\begin{theorem}
  \label{antipodal-points-capture-cone-nonIDD}
  If a pair of antipodal points are included in a connected component of the
escape cone, such a spacetime point is not in an IDD.
  In particular, suppose that each of the escape and capture cones is path-connected. Then, a point is not in an IDD if and only if the escape cone contains a pair of antipodal points.
\end{theorem}

At a spacetime point $p$ in an IDD,
an orthodrome is included in the capture cone.
This immediately implies the existence of an infinite number of pairs
of antipodal points in a connected component of the capture cone.
Then, Thm.~\ref{antipodal-points-capture-cone-ODD}
implies that the point $p$ is also in an ODD, and hence,
we have the following:
\begin{corollary}
  \label{IDD-is-included-in-ODD}
  An IDD is always included in an ODD,
    as long as they are associated with the same $T^\mu$.
\end{corollary}

We now present the useful criteria to find
an ODH and an IDH. 
Assuming the continuity of the mappings from each spacetime point to a configuration of escape and capture cones, 
we can find the following: 
\begin{proposition}
  \label{IDH-ODH-conditions}
  Suppose that each of the escape and capture cones has a single component
at each spacetime point, and there are continuous mappings from each spacetime point to a configuration of escape and capture cones.
  Then, at the ODH, there is an orthodrome 
  that is included in the escape cone except
  at two or more points, which are on the boundary of the escape and capture cones.
  Similarly, at the IDH, there is an orthodrome
  that is entirely included in the capture cone, and furthermore,  
  the orthodrome intersects the boundary of the escape and capture cones at two or more points.
\end{proposition}

There are several remarks on this
Proposition. First, the continuity of the escape and capture cones is guaranteed if the spacetime metric has $C^{2-}$ continuity,\footnote{An example of the spacetime without the continuity of the escape and capture cones is given by a singular hypersurface spacetime in the sense of Israel's definition \cite{Israel:1966}. Note that the ODH and IDH can be defined if the escape cone is uniquely determined, even if the metric is not $C^{2-}$.} and the timelike vector field $T^\mu$ is continuous. 
Here, other than the ODH or the IDH, there might be a surface
having the same property for the ODH/IDH
mentioned in Prop.~\ref{IDH-ODH-conditions}. However, the outermost surface
among them 
can be identified as the ODH/IDH.

\subsection{Properties of Outer Dark Domain and Outer Dark Horizon}
\label{Subsec:PropertyODH}
In this subsection, we provide several properties on the ODD and ODH.
First, since all null geodesics in the Minkowski spacetime reach future null infinity, we have the following proposition. 
\begin{proposition}
  \label{OLM}
    In the Minkowski spacetime, the ODD and ODH are both absent for any associated timelike vector field $T^\mu$.
  \end{proposition}  
  Proposition~\ref{OLM} means that the curved spacetime is essential for the existence of the ODH.

  Next, we focus on the spacetimes with black hole(s).
  For any point $p$ in a black hole, all null geodesics emanating from $p$ do not reach future null infinity by the definition of the black hole. Therefore, $p$ is an element of the ODD.
  This is summarized as the following theorem:
  \begin{theorem}
    \label{OLB}
    In an asymptotically flat spacetime with black hole(s), points in the black hole(s) are also in the ODD for every associated timelike vector field.
  \end{theorem}

  Let us discuss the existence of the ODH in spacetimes with black hole(s) for a particular class of $T^\mu$. To see this, we first show the following lemma.

  \begin{lemma}
    \label{OLP}
    Consider a four-dimensional asymptotically flat spacetime in which the metric near future null infinity behaves as Eqs.~\eqref{metricuu}--\eqref{metricuI} with $C^{2-}$ functions using Bondi coordinates.
    Let a timelike vector $T^\mu$ satisfy $T^\mu \propto (du+dr)^\mu$. We assume that $\Omega_i$, which is defined in Eq.~\eqref{defOi}, is positive. 
    Then, all points $p$ with a sufficiently large $r$ are not in the ODD associated with $T^\mu$. 
  \end{lemma}
  \begin{proof}
     We will show that, for all point $p$ with sufficiently large $r$, there exists a tangent vector $n^\mu$ orthogonal to $T^\mu$, such that all null geodesics which are orthogonal to $n^\mu$ at $p$ reach future null infinity.
     Taking $n^\mu=n_1^\mu:=\left(\partial_u\right)^\mu-\left(\partial_r\right)^\mu$ to be an orthogonal vector to $T^\mu$, we will show that any null vector $k^\mu=\left(\partial_u\right)^\mu+k^r\left(\partial_r\right)^\mu+k^I\left(\partial_I\right)^\mu$~\footnote{$k^u$ is required to be non-zero, otherwise $k^\mu$ becomes spacelike.} orthogonal to $n_1^\mu$ satisfies $k^r=\Cr{-1}$.
  By using Eqs.~\eqref{metricuu}, \eqref{metricur}, and \eqref{metricuI}, the orthogonal condition between $n_1^\mu$ and $k^\mu$, that is, $n_1^\mu k_\mu =0$, gives
  \ba
  -k^r\left(1+\Cr{-2}\right)+k^I\Cr{0}+\Cr{-1}&=&0
  \ea 
  and we have
  \ba
  k^r&=&k^I\Cr{0}+\Cr{-1}.\label{orthogonal_n_k}
  \ea
  Then, a simple calculation gives us 
  \ba
  0&=&k^\mu k_\mu\nonumber\\
  &=&k^Ik^J\left(r^2\omega_{IJ}+\Cr{1}\right)+k^I\Cr{0}-1+\Cr{-1}
\ea
and we see 
\ba
k^I&=&\Cr{-1}.
\ea
With Eq.~\eqref{orthogonal_n_k}, this shows us 
\ba
k^r&=&\Cr{-1}.
\ea
As a consequence, we have  
$k^r=\Cr{-1}=\Cr{-1}k^u$.
By Prop.~\ref{PropIV}, null geodesics with $k^r=\Cr{-1}k^u$ at sufficiently large $r$ region reach future null infinity if $\Omega_i>0$ holds.
Thus, all null geodesics which are orthogonal to $n_1^\mu$ at $p$ reach future null infinity.
Therefore, $p$ is not in the ODD associated with $T^\mu$. 
\end{proof}

Here, Lem.~\ref{OLP} does not depend on whether there exist black hole(s) or not.
  We can also show a higher-dimensional version of Lem.~\ref{OLP} in a similar way. Note that in higher dimensions,  $\Omega_i>0$ is not necessary to show the null geodesics to reach future null infinity~\cite{Amo:2021gcn,Amo:2022tcg}.
  Since the same fact also holds for all other theorems, we will not mention higher-dimensional cases, hereinafter.  

  Now, let us show the existence of the ODH in spacetimes with black hole(s) under the same conditions with Lem.~\ref{OLP}.

\begin{proposition}
    \label{ODHE}
    Consider an asymptotically flat spacetime with black hole(s) under the assumptions in Lem.~\ref{OLP}. The ODH associated with $T^\mu$ is not empty, and not connected to future null infinity.
  \end{proposition}
  
  \begin{proof}
    Let the region of the ODD be denoted by $S_O$.
    Then, $S_O$ is not empty by Thm.~\ref{OLB}. 
    Due to Lem.~\ref{OLP}, a region sufficiently close to future null infinity is a subset of $\mathcal{M}\backslash S_O$. 
  Therefore, the ODH is not empty by Def.~\ref{ODH}, and not connected to future null infinity.
  \end{proof}

  \noindent This proposition is supporting evidence for the ODH to serve as an indicator for a strong gravity region.

\subsection{Properties of Inner Dark Domain and Inner Dark Horizon}
\label{Subsec:PropertyIDH}

In this subsection, we provide several properties on the IDD and IDH.
In the same way to the proofs for Prop.~\ref{OLM} and Thm.~\ref{OLB}, we can easily show the following statements. 

\begin{proposition}
    \label{ILM}
    In the Minkowski spacetime, the IDD and IDH are both absent for any associated timelike vector field $T^\mu$.
  \end{proposition}
  \begin{theorem}
    \label{ILB}
  In an asymptotically flat spacetime with black hole(s), points in the black hole are elements of the IDD for any associated timelike vector field $T^\mu$.
  \end{theorem}
  Proposition~\ref{ILM} means that curved spacetimes are essential for the existence of the IDD and IDH.
  Meanwhile, Thm.~\ref{ILB} guarantees the existence of the IDD in spacetimes with black hole(s).

  To show the existence of the IDH in spacetimes with black hole(s), the following lemma is essential, which can easily be verified by using Cor.~\ref{IDD-is-included-in-ODD} and Lem.~\ref{OLP}.
\begin{lemma}
    \label{ILP}
    Under the assumptions in Lem.~\ref{OLP}, all points $p$ with sufficiently large $r$ are not in the IDD associated with $T^\mu$. 
  \end{lemma}

  One may expect that the IDH would serve as an indicator for the existence of a strong gravity region.
  This is indeed correct as the following proposition. 
  
\begin{proposition}
  \label{IDHE}
  Consider an asymptotically flat spacetime with black hole(s) under the assumptions in Lem.~\ref{OLP}. The IDH associated with $T^\mu$ is not empty, and not connected to future null infinity.
\end{proposition}

\begin{proof}
  The IDD $S_I$ is not empty by Thm.~\ref{ILB}. In addition, a region sufficiently close to future null infinity is a subset of $\mathcal{M}\backslash S_I$ due to Lem.~\ref{ILP}. Therefore, the IDH is not empty, and not connected to future null infinity by Def.~\ref{IDH}.
  \end{proof}

Finally, let us prove that the ODH coincides with the IDH in a spherically symmetric spacetime for the spherically symmetric unit timelike vector
field $T^\mu$:
\begin{theorem}
  \label{ODHsp}
  Consider an asymptotically flat, spherically symmetric spacetime.
  Let $x^I$ be the angular coordinates.
  We adopt the timelike vector field $T^\mu$ to be the spherically
  symmetric one that satisfies $T^\mu (dx^I)_\mu=0$.
Then, the ODH associated with $T^\mu$ coincides with the IDH associated with $T^\mu$.
\end{theorem}
\begin{proof}
  If the ODD and IDD coincide with each other, the ODH and IDH also coincide with each other.
  Here, the IDD is a subset of the ODD by Cor.~\ref{IDD-is-included-in-ODD}. Therefore, the task here is to show that the ODD is a subset of the IDD.

  Let a point $p$ be in the ODD.
  By the definition of the ODD, for all spacelike vectors $n^\mu$ orthogonal to $T^\mu$, there exists a null geodesic which emanates from $p$,
  whose tangent vector $k^\mu$ is orthogonal to $n^\mu$ at $p$,
  such that it will not reach future null infinity.
  We now adopt $n^\mu$ as the one satisfying $n^I=0$.
  Due to the spherical symmetry of the spacetime,
  all null geodesics with the tangent vector $k^\mu$
  satisfying $k^\mu n_\mu=0$ emanating from $p$
  will not reach future null infinity. Then, $p$ satisfies
  the condition of being in the IDD.
\end{proof}

%
%
%======================================%
%<<<<<<<<<<<< SECTION IV  >>>>>>>>>>>>>>%
%======================================%
%
\section{Explicit example I: static and spherically symmetric spacetimes}
\label{Sec:stsp}

To understand properties of the dark horizon in a simple example as a first step, let us focus on a static and spherically symmetric spacetime with the metric;
\ba
ds^2&=&-f(r)dt^2+g(r)dr^2+r^2\left(d\theta^2+\sin^2\theta d\phi^2\right)\label{stsp}
\ea
with
$f(r)=1+\Cr{-1}$ and $g(r)=1+\Cr{-1}$.
To clarify the dependence of associating timelike unit vector $T^\mu$ on the position of the dark horizon, we divide our analyses into two parts.
In subsection~\ref{static_flow}, we assume $T^\mu$ to be proportional to timelike Killing vector, and explain the similarity between the dark horizon and the photon sphere.
In subsection~\ref{dynamical_flow}, we generalize $T^\mu$ to include the radial component, and see its dependence on the dark horizon. Subsections~\ref{static_flow} and \ref{dynamical_flow} assume static and dynamical light source flow setups, respectively.

\subsection{Static Flow}
\label{static_flow}

In this subsection, we take an associating timelike unit vector field $T^\mu$ 
to be $T^\mu=(\partial_t)^\mu/\sqrt{f(r)}$ in Defs.~\ref{ODD}--\ref{IDH}, and consider the ODH and IDH
associated with $t-$constant hypersurfaces. 
This setup corresponds to a simple model assuming that light sources are static, {\it i.e.}, fixed at their own constant radii.
For simplicity, we omit writing ``associated with $T^\mu$'' and ``associated
with $t-$constant hypersurfaces'' here.
Since the ODD and IDD coincide due to Thm.~\ref{ODHsp}, we will simply refer to the ODD and IDD as the dark domain, and the ODH and IDH as the dark horizon, respectively.
Along with the above $T^\mu$, a certain point is in the dark domain if and only if null geodesics passing that point with zero radial velocity do not reach future null infinity.

Let us review a condition for a null geodesic to reach future null infinity in a static and spherically symmetric spacetime~\cite{Ogasawara:2020frt,Carter:1968rr,Walker:1970un}. Due to these symmetries, the impact parameter $b:=-k_\phi/k_t$ is conserved along the geodesic. 
Without loss of generality, we can restrict our considerations on the equatorial
plane with $b\geq0$.
Then, the geodesic equation for $r$ is 
\ba
r^2r'^2=\frac{\left(r^2-fb^2\right)}{fg},\label{rgeo}
\ea
where the prime is the derivative with respect to the affine parameter.
Then, due to the non-negativity of the right-hand side of Eq.~\eqref{rgeo}, $b$ is restricted to $|b|\leq B(r):=r/\sqrt{f(r)}$.
Taking the derivative of Eq.~\eqref{rgeo} gives us
  \ba
    \left.r''\right|_{r'=0}=\frac{1}{r f^{1/2}(r)g(r)}\frac{dB}{dr},\label{r''stsp}
  \ea
  which tells us that the signs of $\left.r''\right|_{r'=0}$ and $dB/dr$ coincide. 
  In particular, $dB/dr=0$ is a necessary and sufficient condition for $\left.r''\right|_{r'=0}=0$, which determines the locus of the photon sphere. 
Since $b$ is conserved along a null geodesic, the behavior of a null geodesic is determined by a $b-$constant line, constrained with the potential barrier of $B(r)$ on the $(r,b)$-plane (see Fig.~\ref{Stasp}), similarly to the fundamental discussion on the effective potential in the classical mechanics. 
For a function $B(r)$ drawn in Fig.~\ref{Stasp}, for instance, the fate of a null geodesic emitted with $r'=0$ ({\it i.e.}, starting from a point on $b=B(r)$ in the figure) is determined as follows.
If starting from $r_{II}>r$ or $r_{\rm min}<r<r_I$ with $r'=0$, the radial coordinate of the null geodesic increases to infinity because there is no potential barrier outside of the initial point. If starting from the other region, the radial coordinate on the null geodesic has an upper bound, and hence that null geodesic does not reach future null infinity.
Then, the dark horizon is located at the boundary of these regions, that is $r=r_{\rm min}, r_{I}$ and $r_{II}$. 

\begin{figure}[htbp]
  \begin{center}
    \vspace{1mm}
  \includegraphics[width=10cm]{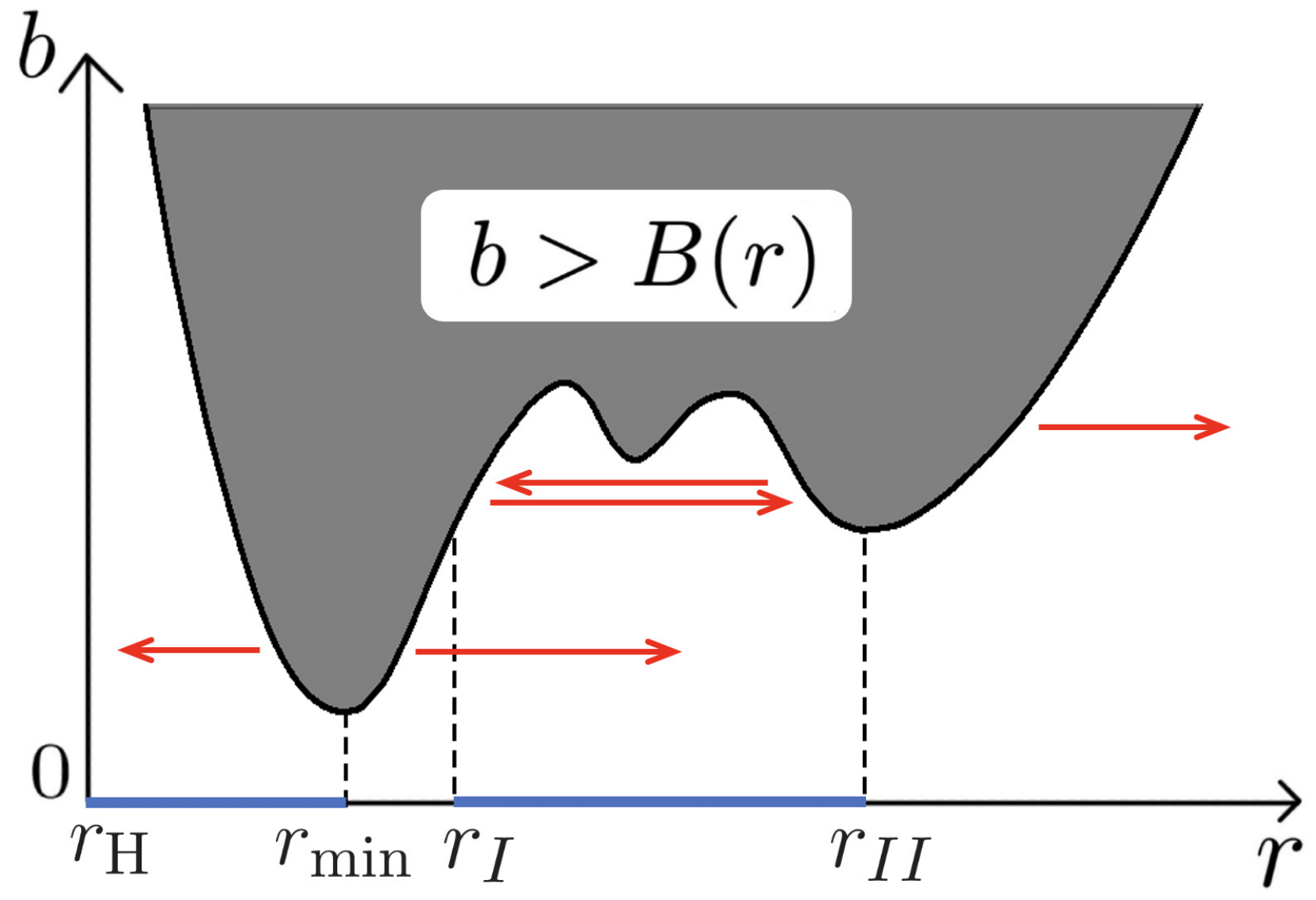}
  \caption{An example of a spacetime having five photon spheres. The shaded region $b> B(r)$ is the forbidden region for a null geodesic. An orbit of null geodesic, described as a $b-$constant line, is reflected if reaching $b=B(r)$ (see red arrows).
  The intervals $[r_{\rm H}, r_{\rm min}]$ and $[r_I, r_{II}]$ correspond to the dark domain, which are indicated by blue lines.
  Here, $r_{\rm H}$, $r_I$ and $r_{II}$ denote the horizon radius, the second outermost dark horizon and the outermost dark horizon, respectively. This spacetime has three dark horizons.}
  \label{Stasp}
  \end{center}
  \end{figure}

  It is easy to see from Fig.~\ref{Stasp} that a ``visible'' photon sphere is a dark horizon.
  Here, we call a photon sphere ``visible'' if null geodesics passing slightly outside of the photon sphere with $r'=0$ reach future null infinity.  
  In particular, the outermost photon sphere, which is visible due to the asymptotic fall-off condition $f(r)=1+\Cr{-1}$ and $g(r)=1+\Cr{-1}$, coincides with the outermost dark horizon.
  Indeed, the visible photon sphere captures the essential properties of the photon sphere---the properties such that it represents the boundary of the observable region, satisfies the areal inequality Eq.~\eqref{AreaInequality} for the outermost visible photon sphere~\cite{Yang:2019zcn}, and serves as an indicator for a strong gravity region.

Note that there are photon spheres that are not dark horizons.
For example, stable photon spheres, at which $B(r)$ is maximal, are not dark horizons.
It should be also noted that there exist dark horizons that are not photon spheres. 
For example, in a spacetime with $B(r)$ of Fig.~\ref{Stasp}, there are two dark domains. The inner boundary of the outer component of the dark domain is a dark horizon, but not a photon sphere.
Let us discuss the physical meaning of this dark horizon.
As shown in Fig.~\ref{Stasp}, we express the radial coordinate of this dark horizon as $r_I$, and that of the outermost dark horizon as $r_{II}$.
Then, we consider the observation of photons emitted from this region. 
For a distant observer, the image
of the observed region can be regarded as a series of concentric circles,
and each circle is specified by the angle $\vartheta$ from the center
of the image. Here, 
the angle $\vartheta$ is a monotonically increasing
function of the impact parameter $b$, and $\vartheta=0$ corresponds to
$b=0$.
When one regards the minimum radial coordinate for each null geodesic as a function of $b$, 
it has a discontinuity at $b_c:=B(r_I)=B(r_{II})$, and hence,
the observable region changes discontinuously at a certain angle
$\vartheta_c$ corresponding
to the impact parameter $b_c$.
Namely, we observe the region
$r>r_{II}$ for $\vartheta=\vartheta_c+\epsilon$ for small $\epsilon$,
while we catch signals from the region
$r\ge r_{I}$ for $\vartheta=\vartheta_c-\epsilon$. 
Due to the above discontinuity of the observable region, the brightness would change discontinuously at $\vartheta=\vartheta_c$.
% This gives the physical meaning of the dark horizon
% that is not the photon sphere. 

\subsection{Dynamical Flow}
\label{dynamical_flow}

In this subsection, let us generalize $T^\mu$ to have the radial component, and study the dependence of $T^\mu$ on the position of the dark horizon.
This situation assumes a distribution of light sources moving in the radial direction. 

We parameterize the direction of $T^\mu$ with the light source's outward velocity $\beta$ ($|\beta|<1$) as
\ba
T^\mu=\frac{1}{\sqrt{1-\beta^2}}\left[\frac{1}{\sqrt{f}}\left(\frac{\partial}{\partial t}\right)^\mu+\frac{\beta}{\sqrt{g}}\left(\frac{\partial}{\partial r}\right)^\mu\right].\label{T}
\ea
If $\beta$ is positive, the light source is moving in the outward direction. We do not require $\beta$ to be constant but impose continuity in the spacetime.
Due to the spherical symmetry of the spacetime, we set $n^\mu$ in Def.~\ref{IDD} to be 
\ba
n^\mu=\frac{1}{\sqrt{f}}\left(\frac{\partial}{\partial t}\right)^\mu+\frac{1}{\beta\sqrt{g}}\left(\frac{\partial}{\partial r}\right)^\mu.
\ea
Consider a null vector $k^\mu$ with $k_t=-1$ which is orthogonal to $n^\mu$.
Without loss of generality, we assume $k^\theta=0$ and consider a point on the equatorial plane in the spacetime.
The $\phi$ component of $k_\mu$ of corresponding to this null vector is 
\ba
k_\phi=\pm r\sqrt{\frac{1}{f}\left(1-\beta^2\right)}.\label{b}
\ea
Then, the condition for a spacetime point to be on the dark horizon is that this null vector is on the boundary of the escape cone at that point.

As a concrete example, let us consider Schwarzschild spacetime, $f=1/g=1-2M/r$. Then Eq.~\eqref{b} reduces to 
\begin{align}
  k_\phi&=\pm r\sqrt{\frac{1-\beta^2}{1-2M/r}}.
 \end{align}
 Here, we recall the fact that the impact parameter $b$ is given by the angular momentum $k_\phi$ divided by the energy $-k_t$, that is, $b=-k_\phi/k_t=k_\phi$. 
 Since the absolute value of the impact parameter corresponding to the boundary of the escape cone is $3\sqrt{3}M$, we consider the condition
 \ba
 3\sqrt{3}M= r\sqrt{\frac{1-\beta^2}{1-2M/r}}
\ea
which gives us the radius of the dark horizon associated with $T^\mu$.
Figure~\ref{Fig:dynamical_T} represents the depedence of $\beta$ on the radial coordinate of the dark horizon $r_{\rm DH}$.
In the case of negative $\beta$, the light source accretes onto the black hole and the dark horizon tends to be larger. Conversely, in the case of positive $\beta$, the light source explodes to distant region, and the dark horizon tends to be smaller.
This corresponds to a situation where photons emitted isotropically in the comoving frame are beamed toward the direction of the source motion.
In other words, the definition of the dark horizon generalizes the photon sphere to include the relativistic beaming effect.

\begin{figure}[htbp]
  \begin{center}
    \vspace{1mm}
  \includegraphics[width=10cm]{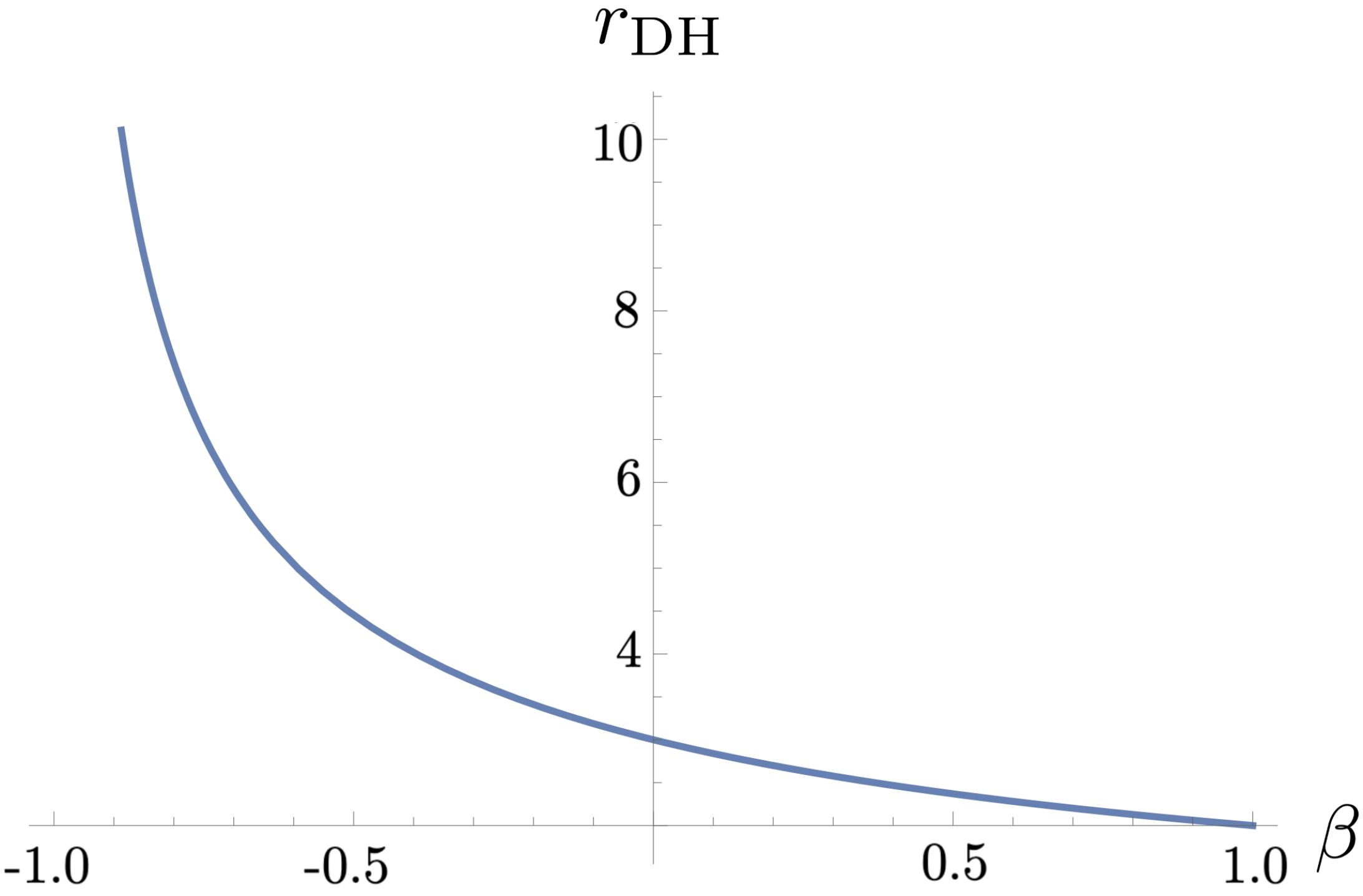}
  \caption{A plot showing the relation between the radial coordinate of the dark horizon ($r_{\rm DH}$) and the value of $\beta$ at that point. The unit of the vertical axis is $M$. The dark horizon is located at $r=3M$ if $\beta=0$ on $r=3M$.
  A point sufficiently close to the horizon can be on the dark horizon if a light source at that point is moving outward nearly at the speed of light ($\beta\simeq1$). In addition, a point with a sufficiently large radial coordinate is on the dark horizon if a light source is moving inward nearly at the speed of light ($\beta\simeq-1$).}
  \label{Fig:dynamical_T}
  \end{center}
  \end{figure}

%
%
%======================================%
%<<<<<<<<<<<< SECTION V  >>>>>>>>>>>>>>%
%======================================%
%
\section{Explicit example II: Vaidya spacetime}
\label{Sec:Vaidya}

In this section, as an example of dynamical spacetimes, we investigate the dark horizon in the Vaidya spacetime.
The Vaidya spacetime has two classes depending on the directions of the flow---ingoing and outgoing Vaidya spacetimes.
Their metrics are written as
\ba
  ds^2 = -\left(1-\frac{2M(v)}{r}\right)dv^2+2dvdr+r^2\left(d\theta^2+\sin^2\theta d\phi^2\right),\label{inVaidya}
\ea
and
\ba
  ds^2 = -\left(1-\frac{2M(u)}{r}\right)du^2-2dudr+r^2\left(d\theta^2+\sin^2\theta d\phi^2\right),\label{outVaidya}
  \ea
  respectively, 
  where $v$ and $u$ are the advanced and retarded times, respectively.
  We focus our attention to the domain where $r>2M(v)$
  and $r>2M(u)$ hold, and assume that the timelike vector field 
$T^\mu$ is proportional to $(\partial_v)^\mu$ and $(\partial_u)^\mu$ in the ingoing and outgoing Vaidya spacetimes, respectively.
Here, we have assumed that the light sources are different from the null fluid matter, and the worldlines of the light sources correspond to $(r,\theta,\phi)$-constant line.
  Due to Thm.~\ref{ODHsp}, the ODH and IDH associated with the
  same unit time vector field $T^\mu$ coincide with each other.
  In what follows, we just call the ODH and IDH the dark horizon, and
do not write ``associated with $T^\mu$'' in this section.

In Refs.~\cite{Koga:2022dsu,Mishra:2019trb}, the shadow edge worldlines in the Vaidya spacetime were discussed for an idealized situation in which the observer and the light source are supposed to be located at future null infinity and past null infinity, respectively.
In this sense, the authors of Refs.~\cite{Koga:2022dsu,Mishra:2019trb} defined the dynamical photon sphere as the boundary of the set of null geodesics connecting future and past null infinity, and 
explicitly calculated the time-dependence of the dynamical photon sphere.
To compare the dark horizon with the dynamical photon sphere, let us examine the dark horizon in the same setup.

In the case with the ingoing Vaidya spacetime, 
we use the same mass function as that in Ref.~\cite{Koga:2022dsu}:
\begin{equation}
  \label{eq:massfunction-cos1}
  M(v)=\left\{
  \begin{array}{cc}
  1 & (v\le 0);\\
  2-\cos\left(\dfrac{v}{v_{\rm f}}\pi\right) & (0<v\le v_{\rm f});\\
  3 & (v>v_{\rm f}),
  \end{array}
  \right. 
\end{equation}
with $v_{\rm f}=100$.
In order to specify the position of
the dark horizon at each $v$, 
we numerically solved the null geodesic equation
with the initial condition
that corresponds to $r=r_0$ and $k^r=0$, where $k^\mu$ is a tangent vector
of the null geodesic. 
By changing the initial values of $r_0$,
we specified the radius of the dark horizon $r_{\rm DH}$ such that the null geodesic
enters the black hole region for $r_0<r_{\rm DH}$ and
escapes to infinity for $r_{\rm DH}<r_0$. 
The dynamical photon sphere is also obtained numerically
in the same way as Ref.~\cite{Koga:2022dsu}.

Figure~\ref{ingoing} shows the radii of the dark horizon, the dynamical photon sphere and the graph $r=3M(v)$.
To see the difference between the dark horizon and the dynamical photon sphere, consider two sets of null geodesics---ones corresponding to the dark horizon and ones corresponding to the dynamical photon sphere. 
The null geodesics emitted from the dark horizon initially have $k^r=0$. 
By contrast, from the behavior of the dynamical photon sphere depicted in Fig.~\ref{ingoing}, the null geodesics that determine the locus of the dynamical photon sphere has $k^r>0$ because the null geodesics are included in the dynamical photon sphere, whose area increases with time. 
This is the reason why the dark horizon is located outside the dynamical photon sphere, shown in Fig.~\ref{ingoing}.
This is a consequence of the assumption that the light sources are located at past null infinity in this case.
Here, both sets of null geodesics describe the boundary of null geodesics which reach future null infinity or not, but
the locations of the light sources are different.

  \vspace{3mm}
\begin{figure}[htbp]
  \begin{center}
  \includegraphics[width=11.5cm]{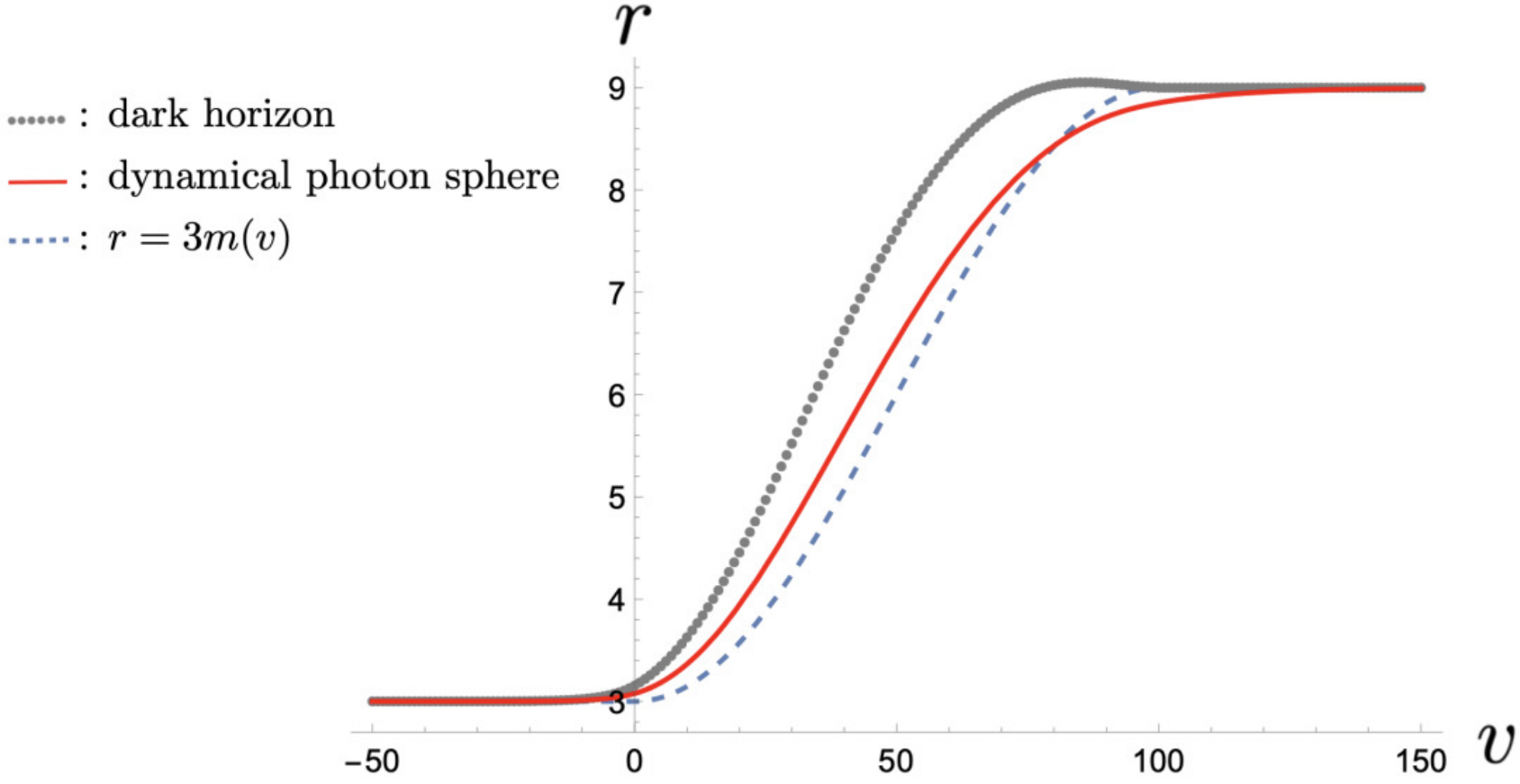}
  \caption{Comparison between the dark horizon (gray dots) and the dynamical photon sphere (red curve) in the ingoing Vaidya spacetime with the metric Eqs.~\eqref{inVaidya} and~\eqref{eq:massfunction-cos1}. The blue dashed curve denotes $r=3M(v)$.
  }
  \label{ingoing}
  \end{center}
  \end{figure}

Similarly, in the case of the outgoing Vaidya spacetime, 
we use the mass function as follows:
\begin{equation}
  \label{eq:massfunction-cos2}
  M(u)=\left\{
  \begin{array}{cc}
  3 & (u\le 0);\\
  2+\cos\left(\dfrac{u}{u_{\rm f}}\pi\right) & (0<u\le u_{\rm f});\\
  1 & (u>u_{\rm f}),
  \end{array}
  \right. 
\end{equation}
with $u_{\rm f}=100$. 
Figure \ref{outgoing} shows the dark horizon, the dynamical photon sphere and the graph $r=3M(u)$.
One can see that the dark horizon is located inside of the dynamical
photon sphere. The reason is interpreted as follows.
  The null geodesics that correspond to the edge of the
  escape cone at the dark horizon has the vanishing radial
  component of the tangent vector, $k^r=0$. 
  By contrast, the plot of the dynamical photon sphere in Fig.~\ref{outgoing} tells us that the null geodesics corresponding to the dynamical photon sphere direct inward because the null geodesics are included in the dynamical photon sphere, whose area decreases with time. 
  This gives us the reason why the dark horizon is located inside of the dynamical photon sphere, shown in Fig.~\ref{outgoing}.
  This is a consequence of the assumption that the light sources are located at past null infinity in this case.
  Here, both sets of null geodesics describe the boundary of null geodesics which reach future null infinity or not, but
  the locations of the light sources are different.
  
  \begin{figure}[htbp]
    \begin{center}
    \includegraphics[width=10.5cm]{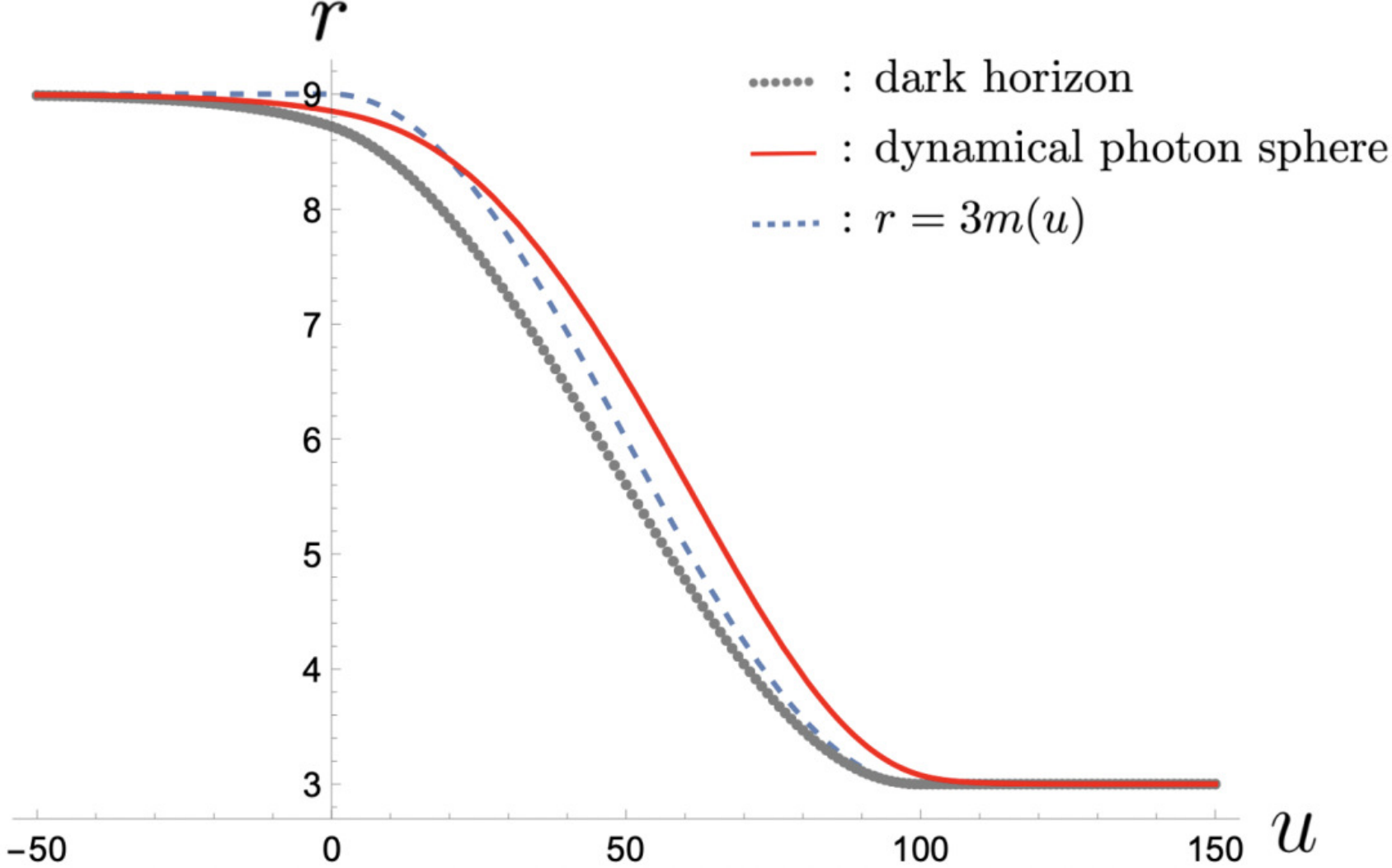}
    \caption{Comparison between the dark horizon (gray dots) and the dynamical photon sphere (red curve) in the outgoing Vaidya spacetime with the metric Eqs.~\eqref{outVaidya} and ~\eqref{eq:massfunction-cos2}. The blue dashed curve corresponds to $r=3M(u)$.
    }
    \label{outgoing}
    \end{center}
    \end{figure}

%
%
%======================================%
%<<<<<<<<<<<< SECTION VI  >>>>>>>>>>>>>>%
%======================================%
%
\section{Explicit example III: Kerr spacetime}
\label{Sec:Kerr}
In this section, we study the dark horizons in the Kerr spacetime as a simple model of typical black holes in our universe, assuming for simplicity that the motion of light sources is the same as that of the zero-angular-momentum observers (ZAMOs).
To be specific, we consider the ODH and IDH associated with a timeslice $\Sigma_t$ of the Boyer-Lindquist coordinates $(t, r, \theta, \phi)$.
This section is organized as follows.
In subsection~\ref{subsec:prepKerr}, we show basic properties of null geodesics in the Kerr spacetime. In subsection~\ref{Subsec:Kerr-axis},
the ODH and IDH on the rotation axis is considered.
In subsections~\ref{Subsec:KerrODD} and \ref{Subsec:KerrIDD}, we investigate the ODH and IDH in the Kerr spacetime, respectively.

\subsection{Preparation}
\label{subsec:prepKerr}
In this subsection, we conduct several calculations as a preparation for the depiction of the shapes of the ODH and IDH in the Kerr spacetime.
The metric of the Kerr spacetime is given by
\ba
\label{Kerrmet}
ds^2
&=&-\frac{\Sigma\Delta}{A}dt^2+\frac{\Sigma}{\Delta}dr^2+\Sigma d\theta^2 +\frac{A}{\Sigma}\sin^2\theta\left(d\varphi-\frac{2Mar}{A} dt\right)^2,
\ea
where $\Sigma$, $\Delta$, and $A$ are defined as
\ba
  \Sigma:=r^2+a^2\cos^2\theta,~~
  \Delta:=r^2-2Mr+a^2,~~A:=\left(r^2+a^2\right)^2-a^2\Delta\sin^2\theta.
  \label{Kerr-metric-definitions}
  \ea
  Here, $M$ is the ADM mass and $a$ is the rotation parameter
  that is related to the ADM angular momentum as $J=Ma$.
  We also sometimes use the dimensionless rotation parameter
  $a_*=a/M$ to specify the solution.
  Below, the black hole spacetime with a nondegenerate event horizon,
$0<|a_*|<1$, 
  is considered (see Sec.~\ref{Sec:stsp} for the case of $a_*=0$). The
  event and Cauchy horizons are located
  at $r=r_{\rm H}^{+}$ and $r=r_{\rm H}^{-}$, respectively,
  where
  \begin{equation}
    r_{\rm H}^{\pm} \, = \, M\pm \sqrt{M^2-a^2},
    \label{Location-horizons}
  \end{equation}
  and we shall limit our attention to 
  the region outside the event horizon, $r>r_{\rm H}^+$.
Let us introduce the tetrad in the Kerr spacetime as follows:
\ba
\left(e_0\right)^\mu&=&\left(\sqrt{\frac{A}{\Delta\Sigma}},0,0,\frac{2Mar}{\sqrt{A\Delta\Sigma}}\right),\label{e0}\\
\left(e_1\right)^\mu&=&\left(0,\sqrt{\frac{\Delta}{\Sigma}},0,0\right),
\label{e1}\\
\left(e_2\right)^\mu&=&\left(0,0,\frac{1}{\sqrt{\Sigma}},0\right),
\label{e2}\\
\left(e_3\right)^\mu&=&\left(0,0,0,\sqrt{\frac{\Sigma}{A}}\frac{1}{\sin\theta}\right).
\label{e3}
\ea
Here, $\left(e_0\right)^\mu$ is orthogonal to the
$t$-constant slice $\Sigma_t$ of the
Boyer-Lindquist coordinates. The ODD, ODH, IDD, and IDH
associated with $\Sigma_t$ in the Def.~\ref{DD_DH_associated_with_time_slice}
are equivalent to those associated with $T^\mu=\left(e_0\right)^\mu$
in the Defs.~\ref{ODD}, \ref{ODH}, \ref{IDD}, and \ref{IDH}.
The timelike vector field $\left(e_0\right)^\mu$
is equivalent to the four-vector fields
of the zero-angular-momentum observers (ZAMOs) (e.g. \cite{Frolov:1998}).
The tetrad basis of Eqs.~\eqref{e0}--\eqref{e3}
gives a natural frame associated with ZAMOs,
and we consider the escape/capture cones in this ZAMO frame. 
We show an illustrative figure
  for the spatial vectors $\left(e_1\right)^\mu$, $\left(e_2\right)^\mu$,
  and $\left(e_3\right)^\mu$ in 
  the left panel of Fig.~\ref{DH_kerr_tetrad_figure}.

\vspace{2.7mm}
\begin{figure}[htbp]
    \begin{center}
    \includegraphics[width=12.5cm]{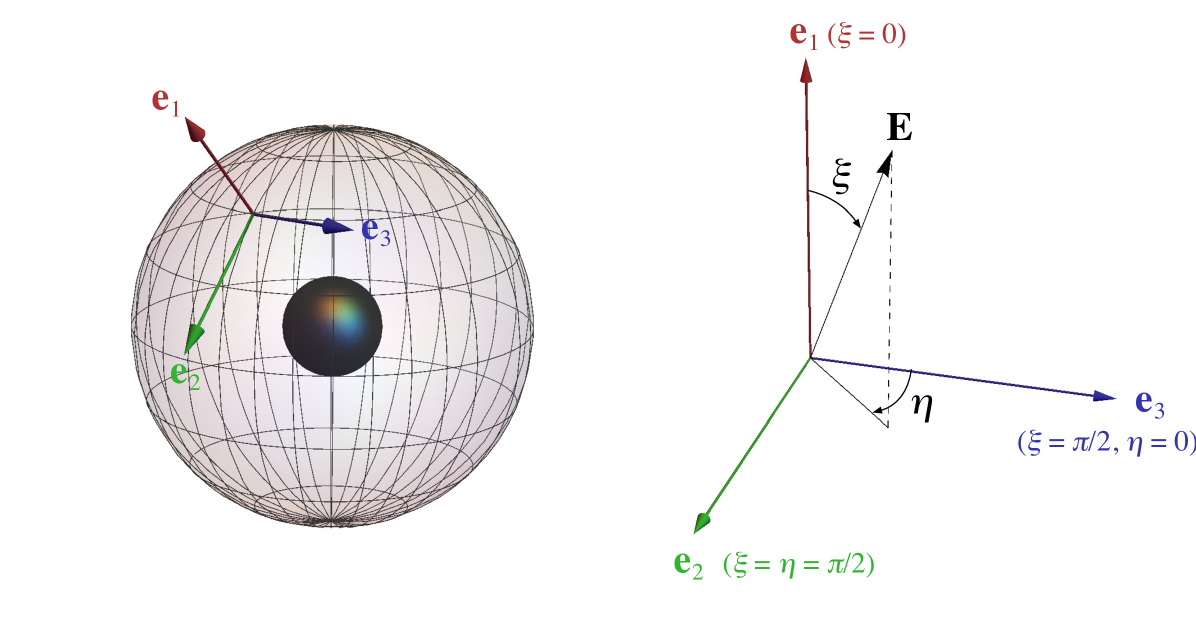}
    \caption{The spatial vectors of the tetrad basis
        and the emission direction vector. Left panel: The configuration
        of $\mathbf{e}_1$, $\mathbf{e}_2$, and $\mathbf{e}_3$ at some spatial point on a $t$-constant hypersurface, where $\mathbf{e}_1$, $\mathbf{e}_2$, and $\mathbf{e}_3$ correspond
        to Eqs. \eqref{e1}, \eqref{e2}, and \eqref{e3}, respectively.
        The black sphere is the event horizon, and the larger sphere indicates an $r$-constant surface.
        At each point, $\mathbf{e}_1$, $\mathbf{e}_2$, and $\mathbf{e}_3$
        are parallel to $\partial_r$, $\partial_\theta$, and $\partial_\phi$, respectively.
        Right panel: The definition of $\xi$ and $\eta$ to parametrize the
        emission direction given in Eq.~\eqref{xieta}.
        Here, $\mathbf{E}$ indicates the emission direction vector given by Eq.~\eqref{emission-direction-vector} [the spatial part of Eq.~\eqref{xieta}]. 
      }
    \label{DH_kerr_tetrad_figure}
    \end{center}
    \end{figure}

At each point $p$ of the Kerr spacetime,
we consider the (initial) null tangent vector $k^\mu$
of a null geodesic that corresponds to a photon emitted
from the point $p$. Adopting the normalization
$k^\mu\left(e_0\right)_\mu=-1$ , the initial null tangent vector
$k^\mu$ is expressed in the form of
\ba
k^\mu&=&\left(e_0\right)^\mu+\cos\xi\left(e_1\right)^\mu+\sin\xi\sin\eta\left(e_2\right)^\mu+\sin\xi\cos\eta\left(e_3\right)^\mu,\label{xieta}
\ea
where $0\le\xi\le\pi$ and $-\pi\le\eta<\pi$.
See the illustration for the definition of $\xi$ and $\eta$
in the right panel of Fig.~\ref{DH_kerr_tetrad_figure}.
Then, at each point $p$, the covariant components $k_t$ and $k_\phi$
are written down as  
\ba
\left.k_t\right|_p&=&-\sqrt{\frac{\Delta\Sigma}{A}}-\frac{2Mar\sin\theta\sin\xi\cos\eta}{\sqrt{A\Sigma}},\\
\left.k_\phi\right|_p&=&\sqrt{\frac{A}{\Sigma}}\sin\theta\sin\xi\cos\eta,\label{bbb}
\ea
where the values of $r$ and $\theta$ at the point $p$ must be used
in these equations.
After the emission, the motion of the photon
follows the null geodesic. In a Kerr spacetime, each null geodesic
possesses
the conserved quantities, the energy $E$, the angular momentum $L$,
and Carter's constant $Q$ \cite{Carter:1968rr}
(we follow the notation of Ref.~\cite{Ogasawara:2020frt}
for Carter's constant).
Their values are determined by the initial condition
({\it i.e.}, the information at the point $p$), and 
for the energy and angular momentum, we have 
$E=\left.-k_t\right|_p$ and $L=\left.k_\phi\right|_p$, respectively.  
The impact parameter $b$ is defined by $b:=L/E$,
and its value becomes
\ba
b\,=\,\left.-\frac{k_\phi}{k_t}\right|_p=\frac{A\sin\theta\sin\xi\cos\eta}{\sqrt{\Delta}\Sigma+2Mar\sin\theta\sin\xi\cos\eta}.\label{impactetaxi}
\ea
In addition, we introduce the dimensionless Carter's constant
$q=Q/E^2$, and the values of $q$ and $b$ satisfy
\ba
q+(b-a)^2&=&\frac{A\sin^2\xi\sin^2\eta+\left[\left(r^2+a^2\right)\sin\xi\cos\eta-a\sqrt{\Delta}\sin\theta\right]^2}{\left(\sqrt{\Delta}+2Mar\sin\theta\sin\xi\cos\eta/\Sigma\right)^2}.\label{qqq}
\ea
Eliminating $\eta$ 
from Eqs.~\eqref{impactetaxi} and \eqref{qqq}, 
one has
\begin{equation}
\sin^2\xi\,=\,\frac{ F(b,q)\Delta}{(A-2Mabr)^2},\label{sinxi-squared}
\end{equation}
where
\ba
F(b,q):=A\left(q+a^2\cos^2\theta\right)+b^2\left[r^4+a^2\cos^2\theta(\Delta-a^2)\right].
\ea
From Eq.~\eqref{sinxi-squared},
it is found that the value of $F(b,q)$ must be nonnegative,
that is, 
the impact parameter $b$ and the dimensionless Carter's constant $q$ can only take values
such that $F(b,q)\ge 0$ is satisfied. 
Taking the square root of Eq.~\eqref{sinxi-squared}
and substituting it into Eq.~\eqref{impactetaxi},
we have
\begin{subequations}
\ba
\sin\xi&=&\frac{\sqrt{ F(b,q)\Delta}}{|A-2Mabr|},\label{sinxi}\\
\cos\eta&=&\mathrm{sgn}(A-2Mabr)\frac{b\Sigma}{\sin\theta\sqrt{F(b,q)}}.\label{coseta}
\ea
\end{subequations}

Let us discuss the relation between $\xi$ and $\eta$ on the boundary of the escape/capture cones.
It is known that the null geodesics corresponding to the boundary of the escape cone neither reach future null infinity nor fall into the black hole.
These null geodesics asymptote to the spherical orbits around the black hole with specific values of $b$ and $q$, which are called the ``spherical photon orbits" \cite{Teo:2020sey}.
To discuss the null geodesics which correspond to the boundary of the escape cone, let us look at the relation between $b$ and $q$ for the spherical photon orbits.

Along the null geodesic, the following first-order differential equation holds~\cite{Carter:1968rr}:
\ba
\Sigma^2 r'^2=R,\label{Kerrgeoeq}
\ea
where 
\ba
R:=\left[E\left(r^2+a^2\right)-La\right]^2-\Delta\left[Q+(L-aE)^2\right].\label{R}
\ea
A spherical photon orbit with the radius $r_s$ satisfies 
$R|_{r=r_s}=0$ and $dR/dr|_{r=r_s}=0$, and these conditions are rewritten as 
\ba
b&=& -\frac{r_s^3-3Mr_s^2+a^2r_s+Ma^2}{a(r_s-M)},
\label{bSPO}\\
q&=& -\frac{r_s^3 \left(r_s^3-6Mr_s^2+9M^2r_s-4a^2M\right)}{a^2\left(r_s-M\right)^2}.\label{qSPO}
\ea
Equations~\eqref{bSPO} and \eqref{qSPO}  give the relation between $b$ and $q$ via a parameter $r_s$ 
that is satisfied by various spherical photon orbits.

Next, consider null geodesics asymptoting to spherical photon orbits in the future ($t\to \infty$). These null geodesics correspond to the boundary of the escape cone at a point on these null geodesics. 
Since $b$ and $q$ of these null geodesics are the same as those of corresponding spherical photon orbits, the same relation of Eqs.~\eqref{bSPO} and \eqref{qSPO} holds.
%null geodesics that asymptote to spherical photon orbits, {\it i.e.} ~\footnote{This step is similar to the discussion of Fig.~\ref{Stasp} in static and spherical symmetric spacetimes except for the fact that we cannot take $q$ to be zero in general in the Kerr spacetime.A null geodesic on the boundary of the escape cone corresponds to the graph where $b$-constant line is tangent to the forbidden region.}.
By substituting Eqs.~\eqref{bSPO} and \eqref{qSPO} into Eqs.~\eqref{sinxi} and \eqref{coseta}, we obtain the relation between $\sin\xi$ and $\cos\eta$
as parametrized by $r_s$, and this relation gives
the boundary of the escape cone.
However, 
since two values of $\xi$ give the same $\sin\xi$ in the
range $0\le \xi\le \pi$, 
the obtained relation
between $\xi$ and $\eta$ 
includes both
the emission directions of photons that asymptote to the
spherical photon orbits in the future ($t\to \infty$) and in the past
($t\to -\infty$). Therefore, we must select an appropriate boundary
to specify the real escape cone.
This can be done by several criteria:
\begin{itemize}
\item[(i)] The structures of the escape cones for
a Schwarzschild spacetime are well known, and the escape cones
change smoothly as the value of $a_*$ is increased;
\item[(ii)] There exist continuous mappings from a spacetime point to a configuration of escape and capture cones;
\item[(iii)] The boundary of the escape/capture cones must be
smooth on the two-sphere of emission directions.
\end{itemize}
We have found that the formula of $\cos\xi$ that satisfies
the above requirements is given by
\begin{equation}
  \cos\xi \ = \ \mathrm{sgn}(r_s-r)\sqrt{1-\frac{ F(b,q)\Delta}{(A-2Mabr)^2}}.
  \label{cos-xi-for-appropriate-boundary}
\end{equation}
In fact, the function in the square root
of Eq.~\eqref{cos-xi-for-appropriate-boundary}
has the factor $(r_s-r)^2$, and hence,
Eq.~\eqref{cos-xi-for-appropriate-boundary}
gives an analytic formula with respect to $r_s$ for any fixed
$r$ and $\theta$. The explicit formula is shown
as Eq.~\eqref{yrs} with Eq.~\eqref{Def-hrs} 
in App.~\ref{Sec:Proof_convex-upward}.

\vspace{2.7mm}
\begin{figure}[htbp]
    \begin{center}
    \includegraphics[width=6.5cm]{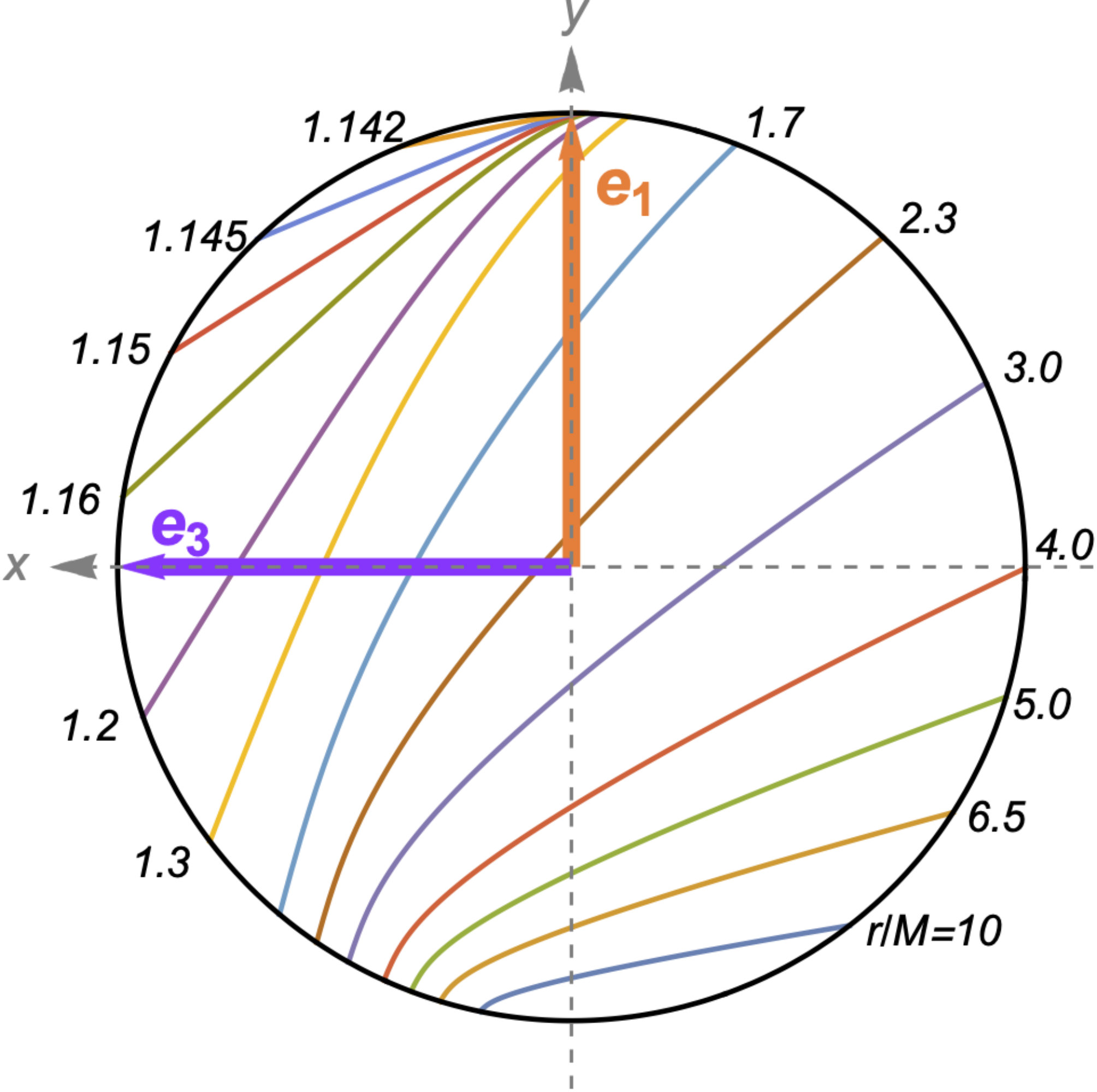}
    \caption{
      The boundaries of the escape and capture cones
      projected onto the plane spanned by $\mathbf{e}_3$ and $\mathbf{e}_1$
      for $r/M=1.142$, $1.145$, $1.15$,
      $1.16$, $1.2$, $1.3$, $1.7$, $2.3$, $3.0$, $4.0$, $5.0$,
      $6.5$, and $10.0$ in the equatorial plane $\theta=\pi/2$ 
      for a Kerr black hole with $a_*=0.99$. 
      The upper region of each curve is the escape cone, while the
      lower region is the capture cone.
      }
    \label{escape}
    \end{center}
    \end{figure}

Here, we present examples of the boundaries of the escape and capture cones.
For this purpose, we consider the projection of the two-sphere
of emission directions onto the plane spanned by $\mathbf{e}_1$
and $\mathbf{e}_3$.
In other words, we consider a ``side view'' of the two-sphere
as seen from a distant position
with $\xi=\eta=\pi/2$ in the tangent subspace.  
Then, the two-sphere is projected
to a unit disk, and the boundary of the escape and
capture cones becomes a curve that has
the endpoints on the unit circle due to the symmetry
in the transformation $\eta\to -\eta$.
Figure~\ref{escape} presents the projected
boundary of the escape and capture cones
for various $r$ values on the equatorial plane $\theta=\pi/2$
for $a_*=0.99$. The upper region of each curve is the escape cone,
while the lower region is the capture cone.
Due to the effect of dragging into rotation of the Kerr black hole,
the escape cone tends to be directed in the $\mathbf{e}_3$ direction.

We point out important properties
of the curve of the projected boundary of the escape and
capture cones. 
Let us introduce the Cartesian coordinates $(x,y)$
on the two-dimensional plane
defined by $\partial_x=\mathbf{e}_3$ and $\partial_y=\mathbf{e}_1$. 
The origin of the coordinates is set so that
the projected two-sphere of emission direction
is given by $x^2+y^2\le 1$. 
Note that $\left(\xi,\eta\right)=\left(\pi/2,\pm \pi/2\right)$ are both projected to the origin.
The curve of the projected boundary
can be expressed as $y=y(x)$. Then,
as proved in App.~\ref{Sec:Proof_convex-upward},
it is possible to show the inequalities
\begin{equation}
  \frac{dy}{dx}\, \le \, 0,
  \label{Eq:negative-slope_boundary}
\end{equation}
and
\begin{equation}
  \frac{d^2y}{dx^2}\, \le \, 0,
  \label{Eq:convex-upward}
\end{equation}
where equalities hold on and only on the rotation axis.
Denoting the two endpoints of the projected boundary
of the escape and capture cones
as $(\xi,\eta)=(\xi_+,0)$ and $(\xi_-,\pi)$, 
the first inequality of Eq.~\eqref{Eq:negative-slope_boundary}
indicates that 
$\xi_+\ge\xi_-$ always holds for $a_*>0$. 
The second inequality of Eq.~\eqref{Eq:convex-upward} means that
the curve of the projected boundary
is convex upward ({\it i.e.} concave downward). This property will be used in
figuring out the ODH and IDH.

For a later convenience, 
let us examine 
the behavior of a null geodesic
whose tangent vector $k^\mu$ of Eq.~\eqref{xieta} has the parameter values
$(\xi,\eta)=(\pi/2,\pm\pi/2)$, corresponding to $(x,y)=(0,0)$, at the emission point $p$
of the corresponding photon.
The motion of such a photon is initially directed toward
$\pm(e_2)^\mu$ in the ZAMO frame and tangent to
the $r$-constant surface, $r^\prime=0$. The photon
possesses zero impact parameter, $b=0$, by Eq.~\eqref{impactetaxi}. 
Due to Eq.~\eqref{Kerrgeoeq}, $r'=0$ is equivalent to $R=0$,
and the value of $dR/dr$ under the condition $R=0$
is computed as 
\ba
\left.\frac{dR}{dr}\right|_{R=0}&=&\frac{2E^2}{\Delta}\left(r^2+a^2\right)\left(r^3-3Mr^2+a^2r+Ma^2\right).
\ea
We can easily see that $\left.{dR}/{dr}\right|_{R=0}=0$
has three real solutions,
and only one solution outside the event horizon,
$r>r_{\rm H}^+$, is given by
\ba
r=r_s^{(0)}:=M+2\left(M^2-\frac{a^2}{3}\right)^{1/2}\cos\left[\frac{1}{3}\cos^{-1}\left\{\frac{M\left(M^2-a^2\right)}{\left(M^2-a^2/3\right)^{3/2}}\right\}\right]\label{SPORadius}.
\ea
This gives the radius of the spherical photon orbit, and the trajectory
of a photon in this orbit crosses the rotation axis due to
the property $b=0$. 
It is easy to see that $\left.dR/dr\right|_{R=0}<0$ holds
in the range $r_{\rm H}^+\le r< r_s^{(0)}$, and the emitted photon falls into the black hole.
We also see $\left.dR/dr\right|_{R=0}>0$
in the range $r> r_s^{(0)}$, which means that $r''>0$
initially.
For such situations, 
$R$ is positive in the region $r> r_s^{(0)}$,
and hence, the emitted photon escapes to future null infinity.

We would like to study one more case: the case
$\sin\xi=0$, corresponding to $(x,y)=(0,\pm 1)$, at the emission point $p$.
In this case, the photon is emitted in the $\pm(e_1)^\mu$
direction in the ZAMO frame.
Equations~\eqref{impactetaxi} and \eqref{qqq}
indicate $b=0$ and $q=\left.-a^2\cos^2\theta\right|_p$. Then,
the function $R$ of Eq.~\eqref{Kerrgeoeq}
is positive and there is no turning point for the radial motion.
Therefore, the photon emitted in the $+(e_1)^\mu$ direction
escapes to infinity, while the photon emitted in the $-(e_1)^\mu$ direction
falls into the black hole.
This means that the points $\xi=0$ and $\xi=\pi$ on the
sphere of emission directions belong to the escape cone
and the capture cone, respectively, 
at an arbitrary spacetime point outside of the horizon.

\subsection{Dark horizons on the rotation axis}
\label{Subsec:Kerr-axis}

Here, we consider the ODH and IDH on the rotation axis.
Strictly speaking, since the tetrad frame of Eqs.~\eqref{e0}--\eqref{e3}
becomes singular at the rotation axis, $\theta=0$ and $\pi$,
the above analysis must be reconsidered
if the initial emission point is on the rotation axis.
However, this case can be 
handled with a minor modification to the case $\theta\neq 0$ or $\pi$.
Let us focus on the upper axis, $\theta=0$, for simplicity. 
The most convenient way would be to consider
the limit $\theta\to 0$ of the tetrad frame
along $\phi=0$ with $t,r=\mathrm{constant}$.
In this case, $(e_1)^\mu$ is the radial unit vector,
and $(e_2)^\mu$ and $(e_3)^\mu$ are the unit vector
directed to $+\theta$ directions
along $\phi=0$ and $\phi=\pi/2$, respectively. 
Then, the results of Sec.~\ref{subsec:prepKerr}
hold also on the rotation axis.

We can also obtain the radial coordinate of the ODH and IDH, analytically.
Because of the axial symmetry, the escape cone
at a point on the rotation axis also becomes
axially symmetric, and hence,
the boundary of the escape cone
is given by $\xi=\mathrm{const}$. 
Because of this property,
the ODH and IDH become degenerate, 
similarly to the spherically symmetric case.
At the ODH and IDH,
the boundary of the escape cone becomes
$\xi=\pi/2$. 
From the results of Sec.~\ref{subsec:prepKerr},
$\xi=\pi/2$ belongs to the capture cone for $r_{\rm H}^+< r< r_s^{(0)}$,
belongs to the escape cone for $r> r_s^{(0)}$,
and becomes the boundary of the escape and capture cones for $r= r_s^{(0)}$.
Therefore, the location of the ODH and IDH on the rotation axis
is $r= r_s^{(0)}$, {\it i.e.} the location of the
spherical photon orbit with $b=0$.

\subsection{Outer Dark Horizon in Kerr spacetime}
\label{Subsec:KerrODD}
In this subsection, we depict the ODH of the Kerr black hole by applying the results in subsection~\ref{subsec:prepKerr}.
In order to find the ODD, we can
use Thm.~\ref{antipodal-points-capture-cone-ODD}
which tells us that if a pair of antipodal points are included
in the capture cone, that position is in the ODD.
As a pair of antipodal points, we focus 
on the points $(\xi,\eta)=(\pi/2,\pm\pi/2)$, that correspond
to the emission directions $\pm (e_2)^\mu$. 
The motion of the corresponding photons
are discussed in Sec.~\ref{subsec:prepKerr}:
In the regions $r_{\rm H}^+<r< r_s^{(0)}$, $r= r_s^{(0)}$, and
$r> r_s^{(0)}$, they fall into the black hole,
propagate along the circular orbits eternally,
and escape to infinity, respectively.
This means that the region
$r_{\rm H}^+<r\le r_s^{(0)}$ is the ODD. 

We now examine whether $r=r_s^{(0)}$ is the ODH or not
by applying Prop.~\ref{IDH-ODH-conditions}.
There exists an orthodrome
that is tangent to the boundary of the escape and capture
cones at $(\xi,\eta)=(\pi/2,\pm\pi/2)$ for the following reason.
Consider the tangent subspace in
which the two-sphere of emission directions exists.
Since the boundary of the 
cones has the symmetry in the transformation $\eta\to-\eta$,
the tangent lines of the boundary at $(\xi,\eta)=(\pi/2,\pm\pi/2)$
are parallel to each other. Then, there exists
a plane on which the two tangent lines exist,
and the intersection between the plane
and the two-sphere of emission directions becomes
the desired orthodrome.
If that orthodrome is included in the escape cone
except at $(\xi,\eta)=(\pi/2,\pm\pi/2)$,
the surface $r=r_s^{(0)}$ is confirmed to be the ODH. 

\vspace{2.7mm}
\begin{figure}[htbp]
    \begin{center}
    \includegraphics[width=6.0cm]{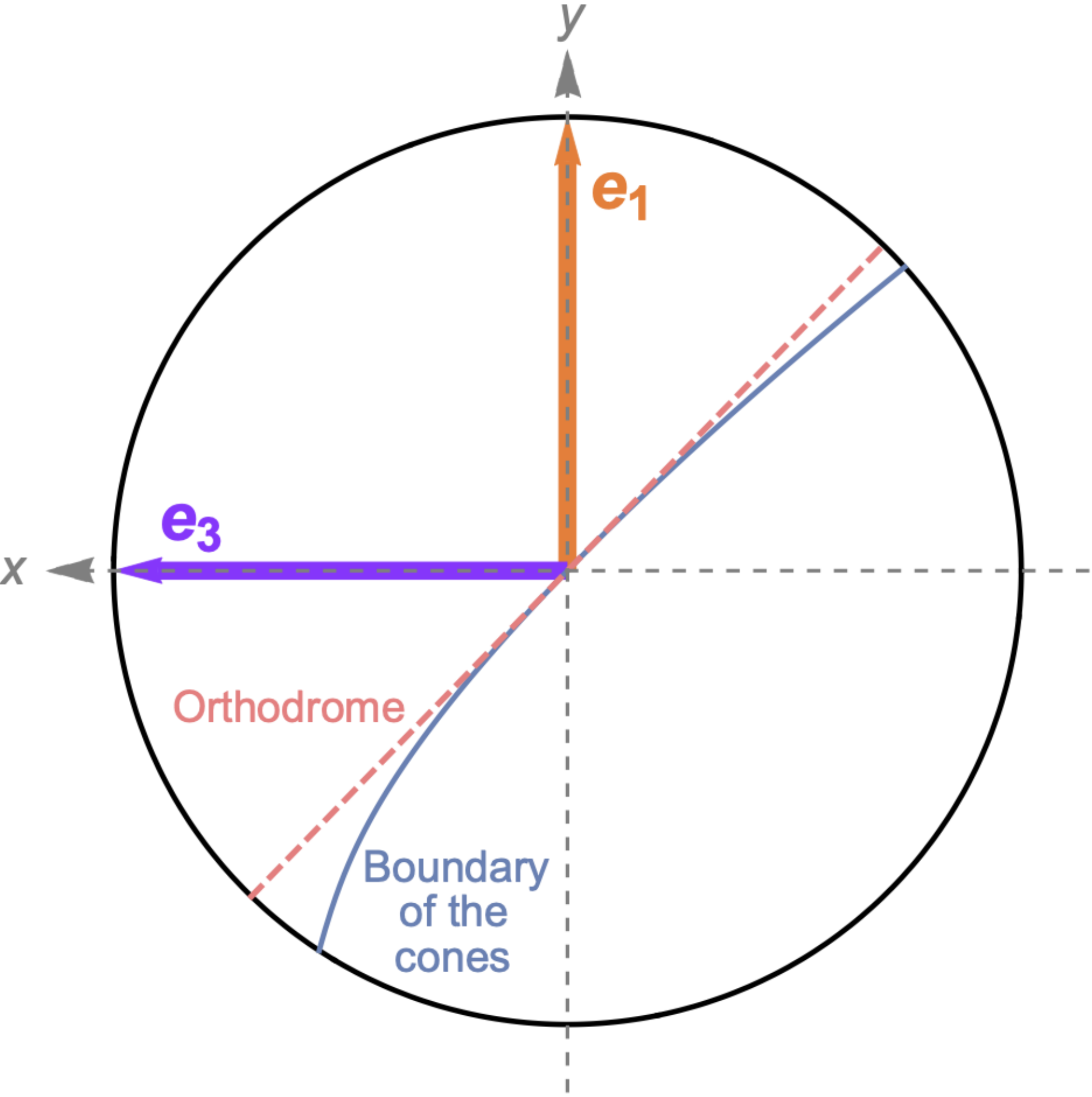}
    \caption{
      The boundary of the escape and capture cones
      projected onto the plane spanned by $\mathbf{e}_3$ and $\mathbf{e}_1$
      for $r/M=r_s^{(0)}/M\approx 2.431$
      in the equatorial plane 
      of a Kerr spacetime with $a_*=0.99$. 
      The projected orthodrome that is tangent to the boundary
      at $(\xi,\eta)=(\pi/2,\pm\pi/2)$ 
      is shown by the dashed line. Because
      the boundary of the cones is convex upward,
      the orthodrome is included in the escape cone
      except at $(\xi,\eta)=(\pi/2,\pm\pi/2)$.
      }
    \label{Boundary_cones_a0990_equatorial-ODH}
    \end{center}
    \end{figure}

Figure~\ref{Boundary_cones_a0990_equatorial-ODH}
shows an example of the projected boundary of the escape and capture cones
and the projected orthodrome
that is tangent to the boundary at $(\xi,\eta)=(\pi/2,\pm\pi/2)$
on the two-dimensional plane spanned by
$\mathbf{e}_3$ and $\mathbf{e}_1$ for $r=r_s^{(0)}$
on the equatorial plane in the case $a_*=0.99$. 
The projected boundary of the cones becomes
a curve that passes through the origin,
and the projected orthodrome becomes a tangent line
to the curve of the projected boundary at the origin. 
The projected orthodrome divides the unit disk
into the upper and lower half-disks, and then, 
the condition that the orthodrome is included in the
escape cone except for $(\xi,\eta)=(\pi/2,\pm\pi/2)$ is
equivalent to that the projected boundary of the cones
is included in the lower half-disk. 
This is guaranteed by the fact that the
projected boundary is convex upward,
as indicated by the inequality of Eq.~\eqref{Eq:convex-upward}.
Therefore, we can safely declare that
$r=r_s^{(0)}$ is the ODH.

\begin{figure}[htbp]
  \begin{minipage}[b]{0.48\linewidth}
      \begin{center}
          \includegraphics[width=6.15cm]{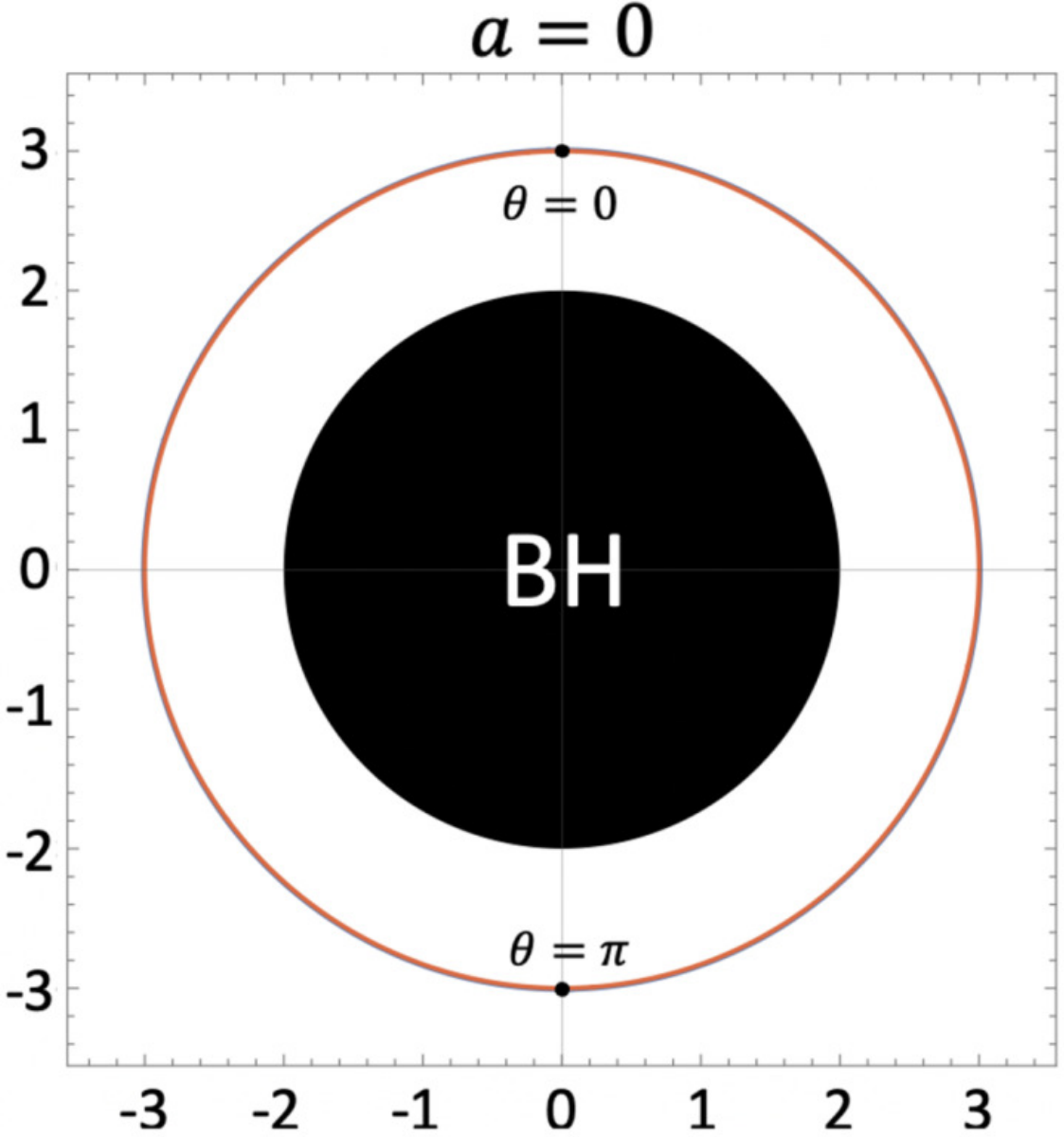}
          \end{center}
  \end{minipage}
  \begin{minipage}[b]{0.02\linewidth}
      ~
  \end{minipage}
  \begin{minipage}[b]{0.48\linewidth}
      \begin{center}
        \vspace{0.7mm}
          \includegraphics[width=6.15cm]{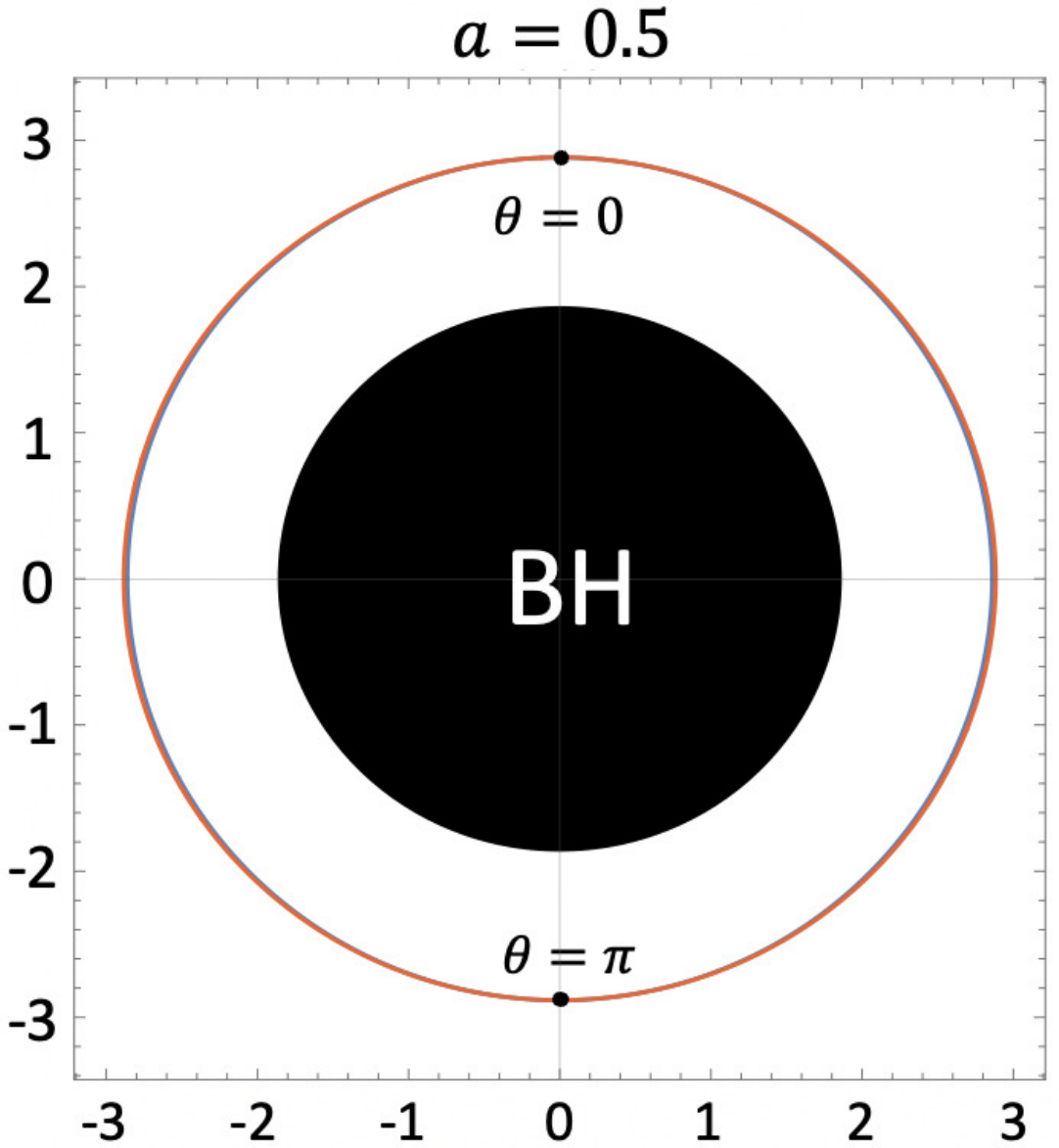}
          \vspace{-0.7mm}
          \end{center}
  \end{minipage}

  \vspace{5.5mm}
  \begin{minipage}[b]{0.48\linewidth}
      \begin{center}
          \includegraphics[width=6.2cm]{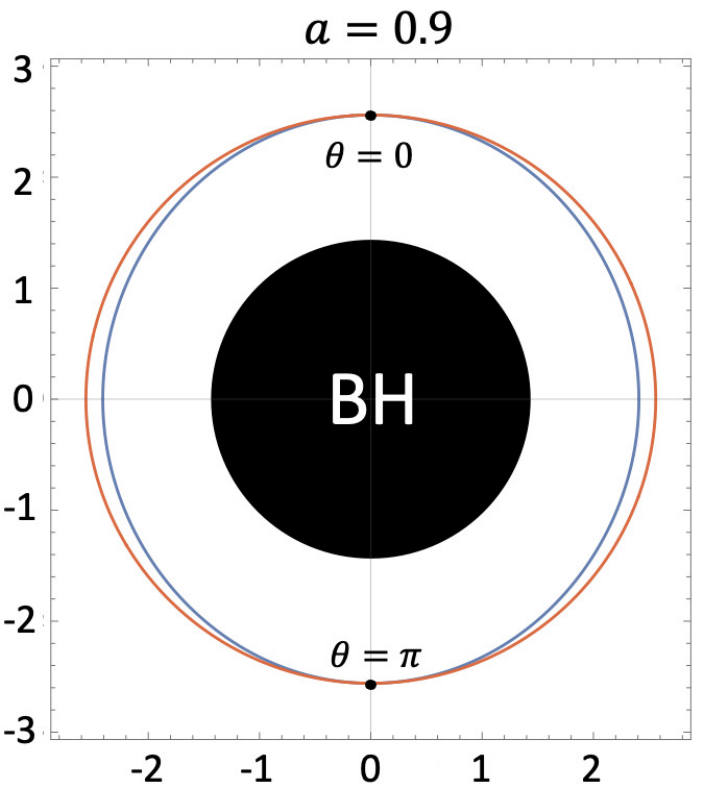}
          \end{center}
  \end{minipage}
  \begin{minipage}[b]{0.02\linewidth}
      ~
  \end{minipage}
  \begin{minipage}[b]{0.48\linewidth}
      \begin{center}
          \includegraphics[width=6.1cm]{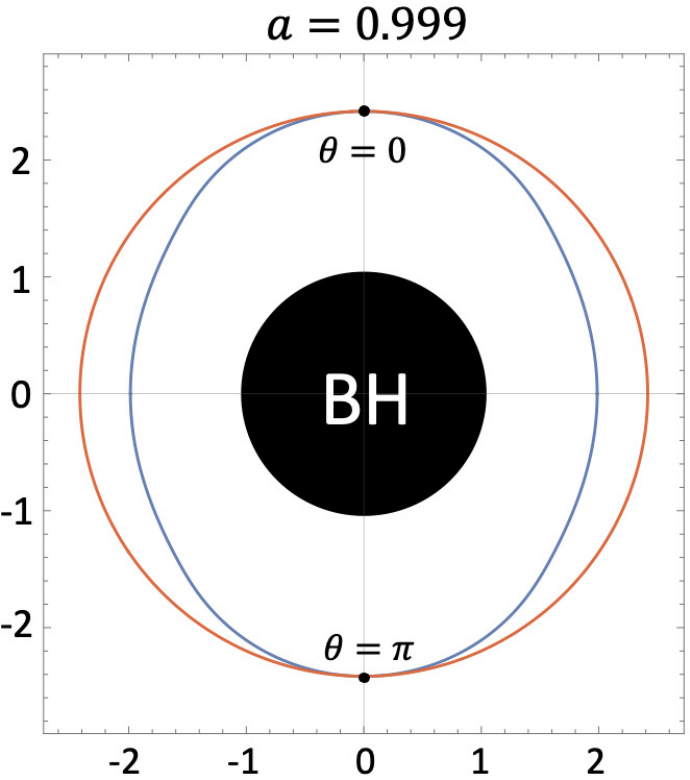}
          \end{center}
  \end{minipage}
  \caption{The outer dark horizons (orange) and inner dark horizons (blue) in the Kerr spacetime seen from a distant position
    on the equatorial plane for $a_*=0,0.5,0.9,0.999$.
    The unit of both axes is $M$.
    The central black disks represent black holes (BHs).
    In the case of $a_*=0$, the outer dark horizon and the inner dark horizon coincide with each other due to Thm.~\ref{ODHsp}.}
    \label{plotDHKerr}
\end{figure}

We depicted the ODH in the Kerr spacetime, {\it i.e.,} $r=r_s^{(0)}$ in Fig.~\ref{plotDHKerr} 
for the cases $a_*=0.0$, $0.5$, $0.9$, and $0.999$. 
In Fig.~\ref{plotDHKerr}, the IDHs, which will be discussed in the next subsection, are shown by blue closed curves.

\subsection{Inner Dark Horizon in Kerr spacetime}
\label{Subsec:KerrIDD}

We now focus our attention to the IDH. The IDH
requires numerical calculations
to specify its location.
As the numerical strategy, 
we make use of
Thm.~\ref{antipodal-points-capture-cone-nonIDD}
and Prop.~\ref{IDH-ODH-conditions},
and in this case, the positions of the 
boundary of the escape and capture cones at $\eta=0$ and $\pi$
play an important role.
Suppose that the position of the boundary
is given by $\xi=\xi_+$ at $\eta=0$ ({\it i.e.}, $\cos\eta=+1$),
and $\xi=\xi_-$  at $\eta=\pi$ ({\it i.e.}, $\cos\eta=-1$).
Then, $\xi_+\le\xi\le\pi$ at $\eta=0$ and $\xi_-\le\xi\le\pi$ at $\eta=\pi$
belong to the capture cone, while $0\le\xi<\xi_+$ at $\eta=0$ and
$0\le\xi<\xi_-$ at $\eta=\pi$ belong to the escape cone. 
For a spacetime point $p$ that satisfies $\xi_-+\xi_+>\pi$,
a pair of antipodal points exists in the escape cone,
such as $(\xi,\eta)=((\xi_+-\xi_-+\pi)/2,0)$
and $((\xi_--\xi_++\pi)/2,\pi)$.
Then, the point $p$ is not in the IDD
from Thm.~\ref{antipodal-points-capture-cone-nonIDD}.

At the distant place $r\gg M$,
the value of $\xi_++\xi_-$ is close to $2\pi$.
As $r$ is decreased, the value of $\xi_++\xi_-$
becomes smaller.
Suppose that a point $p$ satisfying
$\xi_++\xi_-=\pi$ is found.
At that point, we can introduce
the orthodrome that is tangent to the boundary of the escape cone
at $(\xi,\eta)=(\xi_+,0)$ and $(\xi_-,\pi)$
since these two points are antipodal to each other and
the tangent vectors of the boundary of cones at these two points both have vanishing $\xi$ components due to the symmetries in the transformation $\eta\to-\eta$ and in the transformation $\eta-\pi\to-(\eta-\pi)$. 
If that orthodrome is included in the capture cone,
the point $p$ is on the IDH due to Prop.~\ref{IDH-ODH-conditions}.
Therefore, our strategy is firstly to solve for the point
where $\xi_++\xi_-=\pi$, and secondly to
confirm that the obtained point is actually the point of the IDH.

Let us start from solving for 
a position that satisfies $\xi_++\xi_-=\pi$.
We focus on each $\theta$-constant
surface in the Kerr spacetime.
For each $r$, the formulas to parametrically specify the 
boundary of the escape cone are given by Eqs.~\eqref{sinxi}
and \eqref{coseta} with Eqs.~\eqref{bSPO} and \eqref{qSPO}.
These formulas are symbolically written as 
$\sin\xi = \sin\xi(r,r_{\rm s})$ and $\cos\eta = \cos\eta(r,r_{\rm s})$.
Writing the values of $r_s$ that satisfy $\eta=0$ and $\pi$ ({\it i.e.},
$\cos\eta=\pm 1$) as $r_{\rm s}^{\pm}$, we have
$\pm 1  =  \cos\eta(r,r_{\rm s}^{\pm})$.
Then, $\xi_\pm$ are determined by 
$\sin\xi_\pm = \sin \xi(r,r_{\rm s}^{\pm})$.
Since $\xi_++\xi_-=\pi$ leads to the condition
$\sin \xi_+ \ = \ \sin \xi_-$, we obtain the
set of equations to determine the radius $r_{\rm IDH}$ of the IDH
for a given $\theta$:
\begin{subequations}
  \begin{equation}
    \sin \xi(r_{\rm IDH},r_{\rm s}^{+}) \ = \ \sin \xi(r_{\rm IDH},r_{\rm s}^{-}),
    \end{equation}
  \begin{eqnarray}
+1 & = & \cos\eta(r_{\rm IDH},r_{\rm s}^{+}),\\
-1 & = & \cos\eta(r_{\rm IDH},r_{\rm s}^{-}).
  \end{eqnarray}
\end{subequations}
This set of equations can be solved numerically
easily. Except for $\theta=0$ and $\pi$, the value of $r_{\rm IDH}$
is smaller than $r_s^{(0)}$
(interested readers can see blue curves of Fig.~\ref{plotDHKerr}
in advance).

\vspace{2.7mm}
\begin{figure}[htbp]
    \begin{center}
    \includegraphics[width=6.0cm]{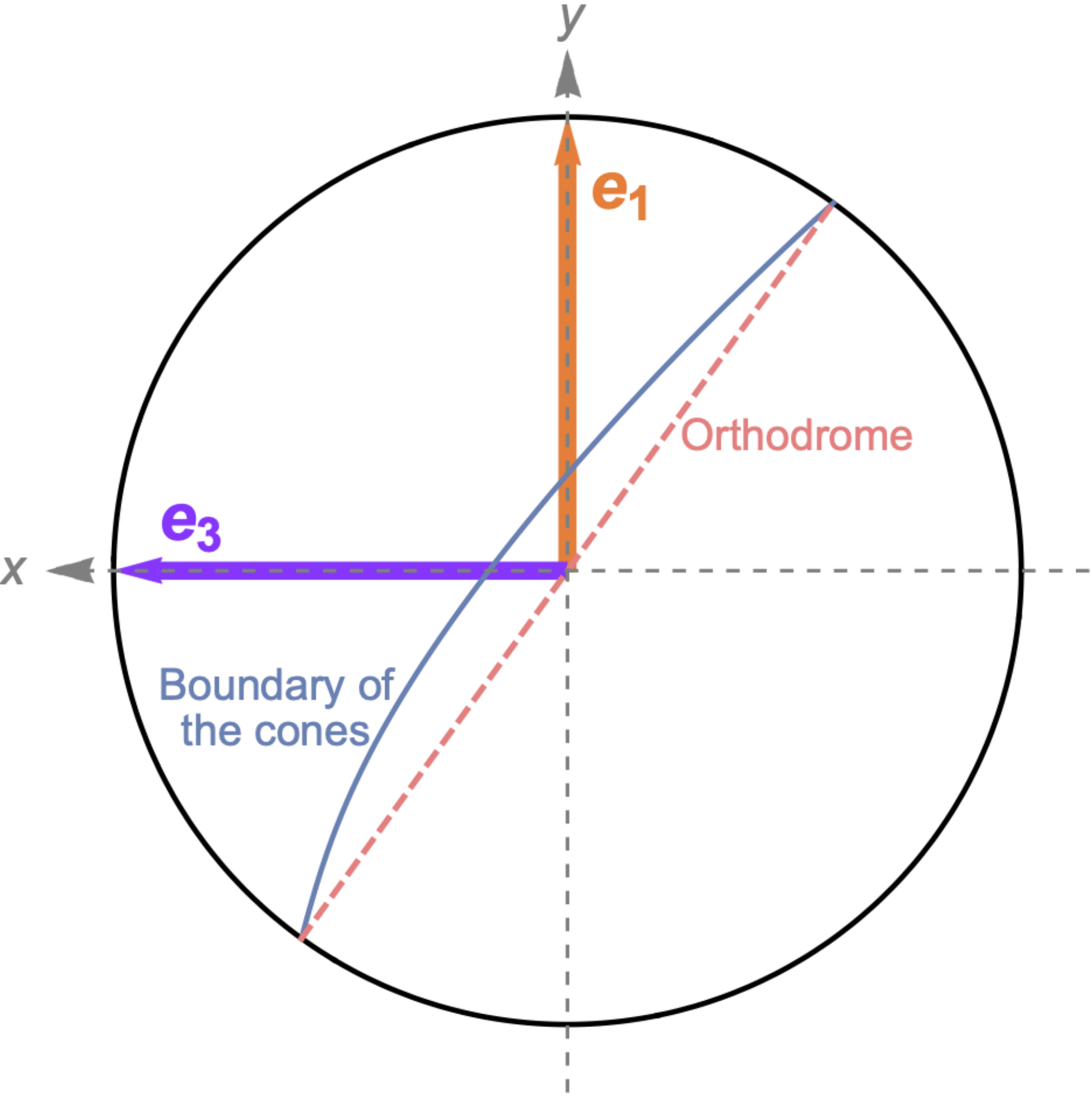}
    \caption{
      The boundary of the escape and capture cones
      projected onto the plane spanned by $\mathbf{e}_3$ and $\mathbf{e}_1$
      for $r/M=r_{\rm IDH}/M\approx 2.092$ 
      in the equatorial plane of a Kerr spacetime with $a_*=0.99$. 
      The orthodrome tangent to the boundary of the cones
      at $(\xi,\eta)=(\xi_+,0)$ and $(\xi_-,\pi)$
      is shown by the dashed line. Because the boundary of the cones
      is convex upward, the orthodrome is included in the
      capture cone.
      }
    \label{Boundary_cones_a0990_equatorial-IDH}
    \end{center}
    \end{figure}

Let us move on to the next step at which we check the obtained surface, $r=r_{\rm IDH}(\theta)$,
is actually the IDH.
Figure~\ref{Boundary_cones_a0990_equatorial-IDH} shows an
example of the boundary of the escape and capture cones and 
the orthodrome tangent to it that are projected onto the
two-dimensional plane spanned by $\mathbf{e}_3$ and $\mathbf{e}_1$ for $r=r_{\rm IDH}$ on the equatorial plane in the case of $a_*=0.99$. 
On the two-dimensional plane, the curve of the projected boundary
and the line of the projected orthodrome have the
same endpoints on the unit circle. 
By the projected orthodrome, the projected
two-sphere of emission directions
is divided into the upper and lower half-disks,
and the condition that the orthodrome is included in the
capture cone 
is equivalent to that the projected boundary of the cones
is included in the upper half-disk.
This is guaranteed by the fact that
the curve of the projected boundary of the cones is
convex upward,
as indicated by the inequality of Eq.~\eqref{Eq:convex-upward}.
Therefore, the obtained numerical solution $r=r_{\rm IDH}(\theta)$
is confirmed to be the IDH. 

In Fig.~\ref{plotDHKerr}, we plot the IDH with blue closed curves
in the Kerr spacetime with several Kerr parameters.
The ODH and IDH both exist in the Kerr spacetime due to Props.~\ref{ODHE} and \ref{IDHE}.
In the case of $a_*=0$, the ODH and IDH coincide with each other by Thm.~\ref{ODHsp}, and their radii are given by $3M$. In the case of $a_*>0$, the IDH coincides with the ODH only on the rotation axis,
and at other positions, 
the ODH is located
in the outside region of the IDH due to Cor.~\ref{IDD-is-included-in-ODD}.
In each case, the ODH is given by $r=r_s^{(0)}$, which coincides with the edge of the shadow observed from infinity on the axis.

%
%
%======================================%
%<<<<<<<<<<<< SECTION VI  >>>>>>>>>>>>>>%
%======================================%
%
\section{Conclusion \& Discussion}

In this paper, we have proposed new concepts dark horizons (the outer dark horizon, ODH for short, and the inner dark horizon, IDH for short) as generalizations of the photon sphere, from static and spherically symmetric spacetimes to general asymptotically flat spacetimes.
 In short, the dark horizons are defined with some timelike vector field $T^\mu$ or the timeslice $\Sigma_t$ which can be chosen along the motion of light sources, and the dark horizons are defined as the boundary of the region where escape (capture) cone is large (small) enough in the frame associated with $T^\mu$ or $\Sigma_t$.
 This definition captures the escape probability of photons emitted from the light sources, and thus, the dependence of the position and the motion of light sources on the brightness observed from infinity.
We have focused on four-dimensional asymptotically flat spacetimes. We can also apply the same definitions in higher dimensions as the dark horizons. In higher-dimensional case, the positive condition for $\Omega_i$, which is defined in Eq.~\eqref{defOi}, is not needed
in Lems.~\ref{OLP}, \ref{ILP}, Props.~\ref{ODHE}, \ref{IDHE}.

We have strictly shown several statements for the dark horizons. First, the ODD and IDD are both absent in the Minkowski spacetime (see Props.~\ref{OLM} and \ref{ILM}) and are both present in the spacetime with black hole(s) under a natural assumption (see Thms.~\ref{OLB} and \ref{ILB}). 
Second, two dark horizons both exist in black hole spacetime under a natural condition (see Props.~\ref{ODHE} and \ref{IDHE}). Third, the IDH is typically located inside of or coinsides with the ODH (see Cor.~\ref{IDD-is-included-in-ODD}).
Fourth, in spherical symmetric spacetimes, the ODH and IDH coincide with each other.
In addition, we have clarified explicit shapes of the  dark horizons in the Vaidya and Kerr spacetimes with simply chosen $T^\mu$. However, we left plots with $T^\mu$ determined by more reasonable motions of light sources for future work.

Through analyses in static and spherically symmetric spacetimes with $T^\mu$ proportional to the timelike Killing field, we have seen that any photon sphere is a dark horizon if null geodesics passing slightly outside of that photon sphere with $r'=0$ reach future null infinity. In this sense, the dark horizon is a generalization of a ``visible'' photon sphere. In addition, if we adopt $T^\mu$ with nonzero radial component that represents the light source moving in radial directions, the dark horizon is located at a different position from the photon sphere. In this sense, the dark horizon is a further generalization of the visible photon spheres to include the relativistic beaming effect.

The dark horizons are defined by referring to null geodesics reaching future null infinity, at which ideal observers are supposed to be distributed.
In this sense, the location of dark horizons are affected not only by the geometry around the dark horizons but also by that in the distant region.
This distinguishes the concepts of the dark horizon from some of other generalizations of the photon sphere defined in terms of the local geometrical quantities, such as the extrinsic curvature (e.g. \cite{Shiromizu:2017ego,Yoshino:2019dty}).

The dark horizons rather resemble the concept of the black room of Ref.~\cite{Siino:2021kep} in the sense that these are defined in terms of the global behavior of the null geodesics.
Here, the black room is defined as a spacetime region such that any photon entering the black room never goes from it.
According to this definition, photons emitted at a point on the boundary of the black room in the tangential
directions to the boundary must either propagate on the boundary of the black room or fall into the black room. 
Then, if we prepare a timelike vector field $T^\mu$ 
which is tangent to the boundary of the maximal black room, the IDD associated with it must include the boundary of the black room.

The main difference between the black room and the dark horizons is as follows.
On the one hand, the condition for a specific point to be in the black room refers to the global behavior of the null geodesics both in the future and past directions of this point, and therefore, it cannot have past endpoints.
Then, the black room does not exist in spacetimes in which a black hole exists only after the gravitational collapse. In addition, the black room does not exist for spacetimes without a black hole~\cite{Siino:2021kep}.
On the other hand, each point of the dark horizons refers to the global behavior of null geodesics only in the future light cone of this point. 
As a result, it is possible for  dark horizons to exist in black hole spacetimes with the gravitational collapse. In addition, it is also possible that the dark horizons exist in spacetimes without a black hole
and they coincide with the visible photon sphere. 
These differences would originate from the difference of the supposed positions of the light sources in these concepts. 
The black room concept would correspond to the situation where light sources are located at past null infinity (as the simple theoretical models of the black hole shadows), whereas the dark horizons correspond to the situation where light sources are located in and around the ODH and IDH, which are typically near the black hole. 
We expect that the dark horizons serve as important concepts to describe the properties of the shadows whose light sources
distribute in strong gravity regions in varieties of asymptotically flat spacetimes.

\bmhead{Acknowledgments}
M.~A. is grateful to Roberto Emparan, Shinji Mukohyama and Takahiro Tanaka for continuous encouragements and valuable suggestions. 
We are grateful to Youka Kaku, Takahiro Tanaka, Shinya Tomizawa, Kazumasa Okabayashi, and Chul-Moon Yoo for useful discussions.
M.~A. is supported by the ANRI Fellowship, JSPS Overseas Challenge Program for Young Researchers and Grant-in-Aid for JSPS Fellows No. 22J20147 and 22KJ1933.
M.~A., K.~I.  and T.~S. are supported by Grant-Aid for Scientific Research from Ministry of Education, Science, Sports 
and Culture of Japan (JP21H05182). K.~I., H.~Y. and T.~S. are also supported by JSPS(No. JP21H05189). 
K.~I. is also supported by 
JSPS Grants-in-Aid for Scientific Research (B) (JP20H01902)
and JSPS Bilateral Joint Research Projects (JSPS-DST collaboration) (JPJSBP120227705). 
T.~S. is also supported by JSPS Grants-in-Aid for Scientific Research (C) (JP21K03551). 
H.~Y. is in part supported by JSPS KAKENHI Grant Numbers JP22H01220,
and is partly supported by Osaka Central Advanced Mathematical Institute 
(MEXT Joint Usage/Research Center on Mathematics and Theoretical Physics JPMXP0619217849).

\begin{appendices}

  \section{Proof of the inequalities of Eqs.~\eqref{Eq:negative-slope_boundary}
    and \eqref{Eq:convex-upward}}
\label{Sec:Proof_convex-upward}

In this appendix, we present the proof of the inequalities
of Eqs.~\eqref{Eq:negative-slope_boundary} and \eqref{Eq:convex-upward}
in the analysis of the Kerr spacetime of Sec.~\ref{Sec:Kerr}.
Before starting, it is useful to discuss the
range of values that $r_s$ can take in Eqs.~\eqref{bSPO}
and \eqref{qSPO}.
Since the range of $r_s$
depends on $\theta$, we regard that the range is given by 
$r_{\rm SPO}^{(+)}(\theta)\le r_s\le r_{\rm SPO}^{(-)}(\theta)$ for a given $\theta$,
where $r=r_{\rm SPO}^{(+)}(\theta)$
and $r=r_{\rm SPO}^{(-)}(\theta)$ correspond to the prograde and retrograde
spherical photon orbits, respectively. 
In the equatorial plane, they take the values (e.g. \cite{Frolov:1998})
\begin{equation}
  r_{\rm SPO}^{(\pm)}(\pi/2) = 2M\left\{
1+\cos\left[\frac23\arccos\left(\mp a_*\right)\right]
  \right\}.
\end{equation}
As the value of $\theta$ is decreased, the values of
$r_{\rm SPO}^{(+)}(\theta)$ increases, while the value of $r_{\rm SPO}^{(-)}(\theta)$
decreases, and both converge to $r_s^{(0)}$
of Eq.~\eqref{SPORadius} in the limit $\theta \to 0$.
In any case, it is sufficient to know that 
the value of $r_s$
satisfies $r_s>r_{\rm H}^+>M$, where $r_{\rm H}^+$ is defined in
Eq.~\eqref{Location-horizons},
since the spherical photon orbits are located outside
the event horizon.

As a preparation, we show that the inequality
\begin{equation}
  A-2Mabr>0
  \label{Inequality:A-2Mabr}
\end{equation}
holds on the boundary of the escape and capture cones,
where $A$ is given in Eq.~\eqref{Kerr-metric-definitions}.
Since the value of $b$ 
is parametrically given by Eq.~\eqref{bSPO} 
on the boundary of the cones, after some algebra, 
we have 
\begin{multline}
  A-2Mabr \, = \,
  2Mr\Delta(r_s) + (r^2+2Mr+a^2\cos^2\theta)\Delta(r)
  \\
  +\frac{2M^2r\left[(r-r_{\rm H}^-)(r_s-r_{\rm H}^+)+(r-r_{\rm H}^+)(r_s-r_{\rm H}^-)\right]}{r_s-M},
\end{multline}
where $\Delta(r)$ is the same as $\Delta$ given in Eq.~\eqref{Kerr-metric-definitions}, $\Delta(r_s)$ is the same as $\Delta(r)$ but $r$ being replaced
by $r_s$, and the definitions of $r_{\rm H}^\pm$
are given in Eq.~\eqref{Location-horizons}.
The right-hand side of this equation
is manifestly positive for $r>r_{\rm H}^+$ and $r_s>r_{\rm H}^+$,
and therefore, the inequality of Eq.~\eqref{Inequality:A-2Mabr}
is satisfied.
This inequality means that
on the boundary of the cones, we do not have to take
an absolute value in the denominator of Eq.~\eqref{sinxi},
and $\mathrm{sgn}(A-2Mabr)$ in Eq.~\eqref{coseta} can be omitted.

We now present the proof of the inequalities
of Eqs.~\eqref{Eq:negative-slope_boundary} and \eqref{Eq:convex-upward}.
In the main article, we introduced the
coordinates $(\xi,\eta)$ on the two-sphere of emission directions
through Eq.~\eqref{xieta},
and in the paragraph with Eq.~\eqref{Eq:negative-slope_boundary},
we introduced the Cartesian coordinates $(x,y)$
on the unit disk that is the two-sphere projected onto the plane spanned by
$\mathbf{e}_3$ and $\mathbf{e}_1$.
The relation between these two coordinate systems is
\begin{subequations}
\begin{eqnarray}
  x &=& \sin\xi\cos\eta,\\
  y &=& \cos\xi.
\end{eqnarray}
\end{subequations}
Through Eqs.~\eqref{sinxi}, \eqref{coseta}, \eqref{bSPO}, and \eqref{qSPO},
the projected boundary of the escape and capture
cones is given in the form $x=x(r_s)$ and $y=y(r_s)$.
More specifically, 
\begin{subequations}
\begin{eqnarray}
  x &=& \frac{b\Sigma\sqrt{\Delta}}{(A-2Mabr)\sin\theta},\\
  y &=& \frac{(r_s-r)\sqrt{H(r_s)A}}{(r_s-M)(A-2Mabr)},
  \label{yrs}
\end{eqnarray}
\end{subequations}
with
\begin{equation}
  H(r_s)\, = \, (r_s-M)^2\left[(r_s+r-2M)^2+4M(r-M)\right]+4M(M^2-a^2)r_s.
  \label{Def-hrs}
\end{equation}
Equation~\eqref{yrs} with Eq.~\eqref{Def-hrs}
gives the explicit form of the right-hand side of
Eq.~\eqref{cos-xi-for-appropriate-boundary},
{\it i.e.}, the analytic expression of $\cos\xi$
for the correct boundary of the cones.

The quantity $dy/dx$ is calculated as
\begin{equation}
  \frac{dy}{dx}\ = \
  \frac{y'}{x'}
  \ = \ 
  -\frac{a\sin\theta}{\Sigma}\sqrt{\frac{\Delta}{A H(r_s)}}
  \left[rg(r_s)+a^2\cos^2\theta h(r_s)\right],\label{Eq:dydx}
\end{equation}
where the prime denotes the derivative with respect to $r_s$ in this appendix, and
\begin{eqnarray}
  g(r_s) & =& (r+M)r_s(r_s+r)+M(3r_s^2+r^2),\\
  h(r_s) &=& (r_s-M)^2 + (r-M)(r_s+M).
\end{eqnarray}
Because $g(r_s)$ and $h(r_s)$ are obviously positive for $r_s>M$ and $r>M$,
the right-hand side of Eq.~\eqref{Eq:dydx} is nonpositive for $a>0$, and 
becomes zero on and only on the rotation axis at which $\sin\theta=0$ is
satisfied~\footnote{Although we have discussed spacetime points off the rotation axis in this subsection, the points on the rotation axis can also be discussed due to the continuity of the escape cones.  See subsection \ref{Subsec:Kerr-axis} for further discussion on the points on the rotation axis.}.
Therefore, the inequality of Eq.~\eqref{Eq:negative-slope_boundary}
has been proved.

The quantity $d^2y/dx^2$ is calculated as
\begin{equation}
  \frac{d^2y}{dx^2}\ = \
  -\frac{Ma^2\sin^2\theta (r+3r_s)(r_s-M)^3(A-2Mabr)^3}
  {\Sigma^2 A^{3/2}K(r_s)H(r_s)^{3/2}},\label{Eq:d^2ydx^2}
\end{equation}
where
\begin{equation}
  K(r_s) \ =\ (r_s-M)^3+M(M^2-a^2).
\end{equation}
Then, the right-hand side of Eq.~\eqref{Eq:d^2ydx^2} is nonpositive
for the black hole configuration $M^2>a^2$ since
$r_s>M$, $r>M$, and the inequality of Eq.~\eqref{Inequality:A-2Mabr} hold,
and becomes zero on and only on the rotation axis.
Therefore, the inequality of Eq.~\eqref{Eq:convex-upward} has been proved.

\backmatter




\end{appendices}



\begin{thebibliography}{9}
    \bibitem{Akiyama:2019cqa}
  K.~Akiyama \textit{et al.} (Event Horizon Telescope Collaboration),
  ``First M87 Event Horizon Telescope results. I. The shadow of the supermassive black hole,''
  Astrophys. J. Lett. \textbf{875}, L1 (2019).
  
  \bibitem{EventHorizonTelescope:2022xnr}
  K.~Akiyama \textit{et al.} [Event Horizon Telescope],
  ``First Sagittarius A* Event Horizon Telescope Results. I. The Shadow of the Supermassive Black Hole in the Center of the Milky Way,''
  Astrophys. J. Lett. \textbf{930}, no.2, L12 (2022).
  
  \bibitem{EventHorizonTelescope:2019pgp}
  K.~Akiyama \textit{et al.} [Event Horizon Telescope],
  ``First M87 Event Horizon Telescope Results. V. Physical Origin of the Asymmetric Ring,''
  Astrophys. J. Lett. \textbf{875}, no.1, L5 (2019).
  
  \bibitem{Wang:2018prk}
  H.~M.~Wang, Y.~M.~Xu and S.~W.~Wei,
  ``Shadows of Kerr-like black holes in a modified gravity theory,''
  JCAP \textbf{03}, 046 (2019).
  
  \bibitem{Moffat:2019uxp}
  J.~W.~Moffat and V.~T.~Toth,
  ``Masses and shadows of the black holes Sagittarius A* and M87* in modified gravity,''
  Phys. Rev. D \textbf{101}, no.2, 024014 (2020).
  
  \bibitem{Khodadi:2020gns}
  M.~Khodadi and E.~N.~Saridakis,
  ``Einstein-\AE{}ther gravity in the light of event horizon telescope observations of M87*,''
  Phys. Dark Univ. \textbf{32}, 100835 (2021).
  
  \bibitem{Claudel:2000} 
  C.~M.~Claudel, K.~S.~Virbhadra, and G.~F.~R.~Ellis,
  ``The Geometry of photon surfaces,''
  J.\ Math.\ Phys.\  {\bf 42}, 818 (2001).
  
  \bibitem{Virbhadra:1999nm}
  K.~S.~Virbhadra and G.~F.~R.~Ellis, ``Schwarzschild black hole lensing,''
  Phys. Rev. D \textbf{62}, 084003  (2000).

  
\bibitem{Gralla:2019xty}
S.~E.~Gralla, D.~E.~Holz and R.~M.~Wald,
``Black Hole Shadows, Photon Rings, and Lensing Rings,''
Phys. Rev. D \textbf{100}, no.2, 024018 (2019).
  
  \bibitem{Yang:2019zcn}
  R.~Q.~Yang and H.~Lu, ``Universal bounds on the size of a black hole,''
  Eur. Phys. J. C \textbf{80} no.10, 949  (2020).
  
  \bibitem{Lu:2019zxb}
  H.~Lu and H.~D.~Lyu, ``Schwarzschild black holes have the largest size,''
  Phys. Rev. D \textbf{101} no.4, 044059  (2020).
    
  \bibitem{Shiromizu:2017ego}
  T.~Shiromizu, Y.~Tomikawa, K.~Izumi and H.~Yoshino,
  ``Area bound for a surface in a strong gravity region,''
  PTEP \textbf{2017}, no.3, 033E01 (2017).
  
  \bibitem{Yoshino:2017gqv}
  H.~Yoshino, K.~Izumi, T.~Shiromizu and Y.~Tomikawa,
  ``Extension of photon surfaces and their area: Static and stationary spacetimes,''
  PTEP \textbf{2017}, no.6, 063E01 (2017).
  
  \bibitem{Cao:2019vlu}
  L.~M.~Cao and Y.~Song,
  ``Quasi-local photon surfaces in general spherically symmetric spacetimes,''
  Eur. Phys. J. C \textbf{81}, 714 (2021).
  
  \bibitem{Yoshino:2019dty}
  H.~Yoshino, K.~Izumi, T.~Shiromizu and Y.~Tomikawa,
  ``Transversely trapping surfaces: Dynamical version,''
  PTEP \textbf{2020}, no.2, 023E02 (2020).
      
  \bibitem{Siino:2019vxh}
  M.~Siino,
  ``Causal concept for black hole shadows,''
  Class. Quant. Grav. \textbf{38}, no.2, 025005 (2021).
    
  \bibitem{Siino:2021kep}
  M.~Siino,
  ``Black hole shadow and wandering null geodesics,''
  Phys. Rev. D \textbf{106}, no.4, 044020 (2022).
  
  \bibitem{Bondi} 
  H.~Bondi, M.~G.~J.~van der Burg and A.~W.~K.~Metzner,
  ``Gravitational waves in general relativity. VII. Waves from axisymmetric isolated systems,''
  Proc. Roy. Soc. Lond. A \textbf{269}, 21-52 (1962).
  
  \bibitem{Sachs} 
  R.~K.~Sachs,
  ``Gravitational waves in general relativity. VIII. Waves in asymptotically flat space-times,''
  Proc. Roy. Soc. Lond. A \textbf{270}, 103-126 (1962).
  
  \bibitem{Tanabe:2011es}
  K.~Tanabe, S.~Kinoshita and T.~Shiromizu,
  ``Asymptotic flatness at null infinity in arbitrary dimensions,''
  Phys. Rev. D \textbf{84}, 044055 (2011).
  
  \bibitem{Hollands:2003xp}
  S.~Hollands and A.~Ishibashi,
  ``Asymptotic flatness at null infinity in higher dimensional gravity,''
  arXiv:hep-th/0311178.
  
  \bibitem{Hollands:2003ie}
  S.~Hollands and A.~Ishibashi,
  ``Asymptotic flatness and Bondi energy in higher dimensional gravity,''
  J. Math. Phys. \textbf{46}, 022503 (2005).
  
  \bibitem{Ishibashi:2007kb}
  A.~Ishibashi,
  ``Higher Dimensional Bondi Energy with a Globally Specified Background Structure,''
  Classical Quantum Gravity \textbf{25}, 165004 (2008).
   
  \bibitem{Amo:2021gcn}
  M.~Amo, K.~Izumi, Y.~Tomikawa, H.~Yoshino and T.~Shiromizu,
  ``Asymptotic behavior of null geodesics near future null infinity: Significance of gravitational waves,''
  Phys. Rev. D \textbf{104} (2021) no.6, 064025 [erratum: Phys. Rev. D \textbf{107}, 029901 (2023)].
  
  \bibitem{Amo:2021rxr}
  M.~Amo, T.~Shiromizu, K.~Izumi, H.~Yoshino and Y.~Tomikawa,
  ``Asymptotic behavior of null geodesics near future null infinity. II. Curvatures, photon surface, and dynamically transversely trapping surface,''
  Phys. Rev. D \textbf{105}, no.6, 064074 (2022).
  
  \bibitem{Amo:2022tcg}
    M.~Amo, K.~Izumi, Y.~Tomikawa, H.~Yoshino and T.~Shiromizu,
    ``Asymptotic behavior of null geodesics near future null infinity. III. Photons towards inward directions,''
    Phys. Rev. D \textbf{106}, no.8, 084007 (2022) [erratum: Phys. Rev. D \textbf{107}, 029902 (2023)].
  
  
  \bibitem{Amo:2023qws}
  M.~Amo, K.~Izumi, Y.~Tomikawa, T.~Shiromizu and H.~Yoshino,
  ``Asymptotic behavior of null geodesics near future null infinity IV: Null-access theorem for generic asymptotically flat spacetime,''
  [arXiv:2305.01767 [gr-qc]].
  
  \bibitem{Israel:1966}
  W.~Israel,
  ``Singular hypersurfaces and thin shells in general relativity,''
  Nuovo Cim. B \textbf{44S10}, 1 (1966)
  [erratum: Nuovo Cim. B \textbf{48}, 463 (1967)].
  
  \bibitem{Ogasawara:2020frt}
  K.~Ogasawara and T.~Igata, ``Complete classification of photon escape in the Kerr black hole spacetime,''
  Phys. Rev. D \textbf{103} (2021) no.4, 044029.
  
  \bibitem{Carter:1968rr}
  B.~Carter, ``Global structure of the Kerr family of gravitational fields,''
  Phys. Rev. \textbf{174} (1968), 1559-1571.
  
  \bibitem{Walker:1970un}
  M.~Walker and R.~Penrose,
  ``On quadratic first integrals of the geodesic equations for type [22] spacetimes,''
  Commun. Math. Phys. \textbf{18}, 265-274 (1970).
  
  \bibitem{Koga:2022dsu}
  Y.~Koga, N.~Asaka, M.~Kimura and K.~Okabayashi,
  ``Dynamical photon sphere and time evolving shadow around black holes with temporal accretion,''
  Phys. Rev. D \textbf{105}, no.10, 104040 (2022).
  
  \bibitem{Mishra:2019trb}
  A.~K.~Mishra, S.~Chakraborty and S.~Sarkar,
  ``Understanding photon sphere and black hole shadow in dynamically evolving spacetimes,''
  Phys. Rev. D \textbf{99}, no.10, 104080 (2019).
  %doi:10.1103/PhysRevD.99.104080
  %[arXiv:1903.06376 [gr-qc]].
  %65 citations counted in INSPIRE as of 31 May 2023
    
  \bibitem{Frolov:1998}
    V.~P.~Frolov and I.~D.~Novikov,
    {\it Black Hole Physics: Basic Concepts and New Developments},
    (Kluwer Academic Publishers, Dordrecht and Boston, 1998).
  
  \bibitem{Teo:2020sey}
  E.~Teo, ``Spherical orbits around a Kerr black hole,''
  Gen. Rel. Grav. \textbf{35}, 1909, (2003).
  
  \end{thebibliography}
\end{document}